\pdfoutput=1
\documentclass[a4paper,UKenglish,cleveref, autoref]{lipics-v2019}
\hideLIPIcs

\usepackage{microtype}

\usepackage{amsmath,amssymb}

\usepackage{xspace}
\usepackage{tikz}
\usepackage{graphicx}
\usepackage{color}
\usepackage{xcolor}
\usepackage{mathrsfs}

\title{Elimination Distances, Blocking Sets, and Kernels for Vertex Cover} 

\author{Eva-Maria C. Hols}{Department of Computer Science, Humboldt-Universit{\"a}t zu Berlin, Germany}{hols@informatik.hu-berlin.de}{https://orcid.org/0000-0002-2832-0722}{Supported by DFG Emmy Noether-grant (KR 4286/1)}
\author{Stefan Kratsch}{Department of Computer Science, Humboldt-Universit{\"a}t zu Berlin, Germany}{kratsch@informatik.hu-berlin.de}{https://orcid.org/0000-0002-0193-7239}{}
\author{Astrid Pieterse}{Eindhoven University of Technology, P.O. Box 513, 5600 MB Eindhoven, The Netherlands}{astridpieterse@outlook.com}{https://orcid.org/0000-0003-3721-6721}{Supported by NWO Gravitation grant ``Networks''}

\authorrunning{E.\,C. Hols, S. Kratsch, and A. Pieterse}

\Copyright{Eva-Maria C. Hols, Stefan Kratsch, Astrid Pieterse}


\category{}

\relatedversion{}

\supplement{}

\nolinenumbers 

\EventEditors{John Q. Open and Joan R. Access}
\EventNoEds{2}
\EventLongTitle{42nd Conference on Very Important Topics (CVIT 2016)}
\EventShortTitle{CVIT 2016}
\EventAcronym{CVIT}
\EventYear{2016}
\EventDate{December 24--27, 2016}
\EventLocation{Little Whinging, United Kingdom}
\EventLogo{}
\SeriesVolume{42}
\ArticleNo{23}

\theoremstyle{plain}
\newtheorem*{thm:theorem6}{Theorem \ref{thm:intro:lp-kernel}}
\newtheorem*{thm:theorem5}{Theorem \ref{thm:intro:hereditary-kernel}}
\newtheorem*{thm:theorem4}{Theorem \ref{thm:intro:precise-bounds}}
\newtheorem*{thm:theorem3}{Theorem \ref{intro:thm:reduce-cc}}
\newtheorem*{thm:theorem1}{Theorem \ref{theorem:lb:intro}}
\newtheorem*{thm:theorem2}{Theorem \ref{theorem:intro:bounded-bs-not-enough}}
\theoremstyle{definition}
\newtheorem{redrule}{Reduction Rule}

\theoremstyle{remark}
\newtheorem{observation}{Observation}

\newcommand{\prob}[1]{\textsc{#1}}

\newcommand{\Oh}{\mathcal{O}}
\newcommand{\VC}{\prob{Vertex Cover}\xspace}
\renewcommand{\P}{\ensuremath{\mathsf{P}}\xspace}
\newcommand{\NP}{\ensuremath{\mathsf{NP}}\xspace}
\newcommand{\RP}{\ensuremath{\mathsf{RP}}\xspace}
\newcommand{\containment}{\ensuremath{\mathsf{NP\subseteq coNP/poly}}\xspace}
\newcommand{\ncontainment}{\ensuremath{\mathsf{NP\nsubseteq coNP/poly}}\xspace}

\newcommand{\OPT}{\ensuremath{\mathrm{OPT}}\xspace}
\newcommand{\LP}{\ensuremath{\mathrm{LP}}\xspace}
\newcommand{\MM}{\ensuremath{\mathrm{MM}}\xspace}
\newcommand{\C}{\ensuremath{\mathcal{C}}\xspace}
\newcommand{\Q}{\ensuremath{\mathcal{Q}}\xspace}
\newcommand{\X}{\ensuremath{\mathcal{X}}\xspace}
\newcommand{\F}{\ensuremath{\mathcal{F}}\xspace}
\newcommand{\td}[1]{\ensuremath{\mathrm{td}(#1)}}
\newcommand{\ed}[2]{\ensuremath{\mathrm{ed}_{#1}(#2)}}
\newcommand{\bsg}[1]{\ensuremath{\beta(#1)}}
\newcommand{\bsc}[1]{\ensuremath{\beta_{#1}}}
\newcommand{\bsd}[2]{\ensuremath{\beta_{#1}(#2)}}

\begin{document}

\maketitle

\begin{abstract}
The \VC problem plays an essential role in the study of polynomial kernelization in parameterized complexity, i.e., the study of provable and efficient preprocessing for \NP-hard problems. Motivated by the great variety of positive and negative results for kernelization for \VC subject to different parameters and graph classes, we seek to unify and generalize them using so-called blocking sets, which have played implicit and explicit roles in many results.

We show that in the most-studied setting, parameterized by the size of a deletion set to a specified graph class \C, bounded minimal blocking set size is necessary but not sufficient to get a polynomial kernelization. Under mild technical assumptions, bounded minimal blocking set size is showed to allow an essentially tight efficient reduction in the number of connected components.

We then determine the exact maximum size of minimal blocking sets for graphs of bounded elimination distance to any hereditary class \C, including the case of graphs of bounded treedepth. We get similar but not tight bounds for certain non-hereditary classes \C, including the class $\C_{\LP}$ of graphs where integral and fractional vertex cover size coincide. These bounds allow us to derive polynomial kernels for \VC parameterized by the size of a deletion set to graphs of bounded elimination distance to, e.g., forest, bipartite, or $\C_{\LP}$ graphs.
\end{abstract}

\section{Introduction}

In the \VC problem we are given an undirected graph $G=(V,E)$ and an integer~$k$; the question is whether there exists a set $S\subseteq V$ of at most $k$ vertices such that each edge of $G$ is incident with a vertex of $S$, or, in other words, such that $G-S$ is an independent set. Despite \VC being \NP-complete, it is known that there are efficient preprocessing algorithms that reduce any instance $(G,k)$ to an equivalent instance with $\Oh(k^2)$ or even at most $2k$ vertices (and size polynomial in $k$). The existence or non-existence of such preprocessing routines for \NP-hard problems has been studied intensively in the field of parameterized complexity under the term \emph{polynomial kernelization},\footnote{A polynomial kernelization is an efficient algorithm that given any instance with parameter value $\ell$ returns an equivalent instance of size polynomial in $\ell$. In the initial \VC example we have $\ell=k$ but many other parameters will be considered. Full definitions can be found in Section~\ref{section::preliminaries}.} and \VC has turned out to be one of the most fruitful research subjects with a variety of upper and (conditional) lower bounds subject to different parameters (see, e.g.,~\cite{BougeretS17}).

In the present work, we seek to unify and generalize existing results by using so-called \emph{blocking sets}. A blocking set in a graph $G=(V,E)$ is any set $Y\subseteq V$ that is not a subset of any minimum cardinality vertex cover of $G $, e.g., the set $V$ itself is always a blocking set. Of particular interest are \emph{minimal} blocking sets, i.e., those that are minimal under set inclusion. Several graph classes have constant upper bounds on the size of minimal blocking sets, e.g., in any forest (or even in any bipartite graph) every minimal blocking set has size at most two. On the other hand, even restrictive classes like outerplanar graphs have unbounded minimal blocking set size, i.e., for each $d$ there is a graph in the class with a minimal blocking set of size greater than $d$. As a final example, cliques are the unique graphs for which $V$ is the only (minimal) blocking set because all optimal vertex covers have form $V\setminus \{v\}$ for any $v\in V$; in particular, any graph class containing all cliques has unbounded minimal blocking set size.

For ease of reading, the introduction focuses mostly on hereditary graph classes, i.e., those closed under vertex deletion; the later sections give more general results modulo suitable technical conditions, where appropriate. Apart from that, throughout the paper, we restrict study to graph classes $\C$ that are \emph{robust}, that is, they are closed under disjoint union and under deletion of connected components. In other words, a graph $G$ is in $\C$ if and only if all of its connected components belong to $\C$. Most graph classes studied in the context of kernels for \VC are robust, and a large number of them are also hereditary. A particular non-hereditary graph class of interest for us is the class $\C_{\LP}$ of graphs whose minimum vertex cover size equals the minimum size of a fractional vertex cover (denoted $\LP$ as it is also the optimum solution value for the vertex cover LP relaxation).

\subsection{Blocking sets and kernels for Vertex Cover}

Most known polynomial kernelizations for \VC are for parameterization by the vertex deletion distance to some fixed hereditary graph class $\C$ that is also robust, e.g., for $\C$ being the class of forests~\cite{JansenB13}, graphs of maximum degree one or two~\cite{MajumdarRS18}, pseudoforests~\cite{FominS16} (each component has at most one cycle), bipartite graphs~\cite{KratschW12}, $d$-quasi forests/bipartite graphs~\cite{HolsK17} (at most $d$ vertex deletions per component away from being a forest/bipartite), cluster graphs of bounded clique size~\cite{MajumdarRS18}, or graphs of bounded treedepth~\cite{BougeretS17}. Concretely, the input is of form $(G,k,X)$, asking whether $G$ has a vertex cover of size at most $k$, where $X\subseteq V$ such that $G-X\in\C$; the size $\ell=|X|$ of the \emph{modulator} $X$ is the parameter. Blocking sets have been implicitly or explicitly used for most of these results and we point out that all the mentioned classes have bounded minimal blocking set size.

As our first result, we show that this is not a coincidence: If $\C$ is closed under disjoint union (or, more strongly, if $\C$ is robust) then bounded size of minimal blocking sets in graphs of $\C$ is necessary for a polynomial kernel to exist (Section~\ref{section::kernelimpliesbmbss}). Moreover, the maximum size of minimal blocking sets in $\C$ yields a lower bound for the possible kernel size.

\newcommand{\theoremone}{Let $\C$ be a graph class that is closed under disjoint union. If $\C$ contains any graph with a minimal blocking set of size $d$ then \VC parameterized by the size of a modulator $X$ to $\C$ does not have a kernelization of size $\Oh(|X|^{d-\varepsilon})$ for any $\varepsilon>0$ unless \containment and the polynomial hierarchy collapses.}
\begin{theorem}\label{theorem:lb:intro}
\theoremone
\end{theorem}

To the best of our knowledge, this theorem captures all known kernel lower bounds for \VC parameterized by deletion distance to any union-closed graph class $\C$, e.g., ruling out polynomial kernels for $\C$ being the class of mock forests (each vertex is in at most one cycle)~\cite{FominS16}, outerplanar graphs~\cite{Jansen13Thesis}, or any class containing all cliques~\cite{BodlaenderJK14}; and getting kernel size lower bounds for graphs of bounded treedepth~\cite{BougeretS17} or cluster graphs of bounded clique size~\cite{MajumdarRS18}. To get lower bounds of this type, it now suffices to prove (or observe) that $\C$ has large or even unbounded minimal blocking set size.

It is natural to ask whether the converse holds, i.e., whether a bound on the minimal blocking set size directly implies the existence of a polynomial kernelization. Unfortunately, we show that this does not hold in a strong sense: There is a class $\C$ such that all graphs in $\C$ have minimal blocking sets of size one, but there is no polynomial kernelization (Section~\ref{section::bmbssisnotsufficient}). More strongly, solving \VC on $\C$ is not in $\RP\supseteq\P$ unless $\NP=\RP$.

\newcommand{\theoremtwo}{There exists a graph class \C such that all graphs in $\C$ have minimal blocking set size one and such that \VC on \C is not solvable in polynomial time (in fact, not in \RP), unless $\mathsf{NP = RP}$.}
\begin{theorem}\label{theorem:intro:bounded-bs-not-enough}
\theoremtwo
\end{theorem}

In light of this result, one could ask what further assumptions on \C, apart from the necessity of bounded minimal blocking set size, are required to allow for polynomial kernels. Clearly, polynomial-time solvability of \VC on the class \C is necessary and (as we implicitly showed) not implied by \C having bounded blocking set size. If, slightly stronger, we require that blocking sets in graphs of $\C$ can be efficiently recognized\footnote{This condition clearly holds for all hereditary classes \C on which \VC can be solved in polynomial time: Given $G=(V,E)$ and $Y\subseteq V$ it suffices to compute solutions for $G$ and $G-Y$. Clearly, the set $Y$ is a blocking set if and only if $\OPT(G)<\OPT(G-Y)+|Y|$.} then we show that there is an efficient algorithm that reduces the number of components of $G-X$ for any instance $(G,k,X)$ of \VC parameterized by deletion distance to \C to $\Oh(|X|^d)$ (Section~\ref{subsec:reducing-num-components}).
This is a standard opening step for kernelization and can be followed up by shrinking and bounding the size of those components (unless bounds follow directly from \C such as for cluster graphs with bounded clique size). Note that this requires that deletion of any component yields a graph in \C (e.g., implied by \C being robust), which here is covered already by \C being hereditary.

\newcommand{\theoremthree}{Let \C be any hereditary graph class with minimal blocking set size $d$ on which \VC can be solved in polynomial time. There is an efficient algorithm that given $(G,k,X)$ such that $G-X\in \C$ returns an equivalent instance $(G',k',X)$ such that $G'-X\in\C$ has at most $\Oh(|X|^d)$ connected components.}
\begin{theorem}\label{intro:thm:reduce-cc}
\theoremthree
\end{theorem}

We point out that the number $\Oh(|X|^d)$ of components is essentially tight (assuming \ncontainment) because the lower bound underlying Theorem~\ref{theorem:lb:intro} creates instances where components have a constant $c=c(d)$ many vertices. Reducing to $\Oh(|X|^{d-\varepsilon})$ components, for any $\varepsilon>0$, would violate the kernel size lower bound.

\subsection{Minimal blocking set size relative to elimination distances}

Recently, Bougeret and Sau~\cite{BougeretS17} presented a polynomial kernelization for \VC parameterized by the size of a modulator $X$ such that $G-X$ has treedepth at most $d$; here $d$ is a fixed constant and the degree of the polynomial in the kernel size depends exponentially on $d$. To get the kernelization, they prove (in different but equivalent terms) that the size of any minimal blocking set in a graph of treedepth $d$ is at most $2^d$, and they give a lower bound of $2^{d-3}$. As our first result here, we determine the exact maximum size of minimal blocking sets in graphs of treedepth $d$ (see below, and see Section~\ref{section:bounded-blocking-sets} for all these results).

Bulian and Dawar~\cite{BulianD16} introduced the notion of elimination distance to a class \C, generalizing treedepth, which corresponds to elimination to the empty graph (see Section~\ref{section::preliminaries} for a formal definition). This is defined in the same way as treedepth except that all graphs from \C get value $0$ rather than just the empty graph. (Note that here it is convenient that \C is robust because the definition assigns value $0$ to disjoint unions of graphs from $\C$.) Intuitively, elimination distance to \C can be pictured as having a tree-like deletion of vertices (as for treedepth) but being allowed to stop when the remaining connected components belong to \C (rather than continuing to the empty graph). For hereditary $\C$, we determine the exact maximum size of minimal blocking sets in graphs of elimination distance at most $d$ to $\C$, denoted $\bsd{\C}{d}$, depending on the maximum minimal blocking set size $\bsc{\C}$ in the class \C.

\newcommand{\theoremfour}{Let $\C$ be a robust hereditary graph class where \bsc{\C} is bounded.
    For every integer $d\geq 1$ it holds that
    \begin{align*}
        \bsd{\C}{d} = \begin{cases}
                            2^{d-1}+1 & \text{, if } \bsc{\C} = 1 \text{,} \\
                            (\bsc{\C}-1) 2^d+1 & \text{, if } \bsc{\C} \geq 2.
                        \end{cases}
    \end{align*}}
\begin{theorem}\label{thm:intro:precise-bounds}
    \theoremfour
\end{theorem}

The lower bound holds as well for any non-hereditary class \C but we only get a slightly weaker upper bound for such classes (and also require a further technical condition called $f$-robustness). In particular, we get such an upper bound for the class $\C_{\LP}$ mentioned above. Note that if \C has unbounded minimal blocking set size then the same is true for graphs of any bounded elimination distance to \C (irrespective of \C being hereditary or not).

The bound for graphs of treedepth at most $d$ is included in the theorem by using that having treedepth at most $d$ is equivalent to having elimination distance at most $d-1$ to the class of independent sets (i.e., graphs of treedepth one), for which all minimal blocking sets have size $1$. Concretely, for treedepth $d$ we get $\beta(d)=1$ for $d=1$ and $\beta(d)=2^{d-2}+1$ for $d\geq 2$.

\subsection{Some consequences for kernels for Vertex Cover}

Our bounds for the minimal blocking set size relative to elimination distances allow us to generalize and combine previous polynomial kernelization results for \VC. We state this explicitly for elimination distances to hereditary graph classes.

\newcommand{\theoremfive}{Let $\C$ be a hereditary and robust graph class for which $\bsc{\C}$ is bounded, such that \VC has a (randomized) polynomial kernelization parameterized by the size of a modulator to $\C$. Then \VC also has a (randomized) polynomial kernelization parameterized by the size of a modulator to graphs of bounded elimination distance to $\C$.}
\begin{theorem}\label{thm:intro:hereditary-kernel}
\theoremfive
\end{theorem}

As an example, this combines known polynomial kernels relative (to the size of) modulators to a forest~\cite{JansenB13} resp.\ to graphs of bounded treedepth~\cite{BougeretS17} to polynomial kernels relative to a modulator to graphs of bounded forest elimination distance. Similarly, the randomized polynomial kernel for \VC parameterized by a modulator to bipartite graphs is generalized to a modulator to graphs of bounded bipartite elimination distance. The approach to this result (also in the non-hereditary case) uses our bounds for minimal blocking set size relative to elimination distances and, apart from that, is inspired by the result of Bougeret and Sau~\cite{BougeretS17}. Intuitively, these kernels are obtained by suitable reductions to the known kernelizable cases, and thus carry over their properties (e.g., being deterministic or randomized).

As an explicit example for the non-hereditary case, we state a new kernelization result relative to the size of a modulator to the class of graphs of bounded elimination distance to $\C_{\LP}$, i.e., bounded elimination distance to graphs where optimum vertex cover size equals optimum fractional vertex cover size.

\newcommand{\theoremsix}{\VC admits a randomized polynomial kernel parameterized by the size of a modulator to graphs that have bounded elimination distance to $\C_{\LP}$.}
\begin{theorem}\label{thm:intro:lp-kernel}
\theoremsix
\end{theorem}

This result subsumes several polynomial kernelizations for \VC (except for their size bounds).
\section{Preliminaries}\label{section::preliminaries}
For $n \in \mathbb{N}$ we use $[n]$ to denote $\{1,\ldots,n\}$.
\subparagraph*{Graphs}
A graph $G$ consists of a vertex set $V(G)$ and an edge set $E(G)$. All graphs considered in this paper are simple and undirected. For a vertex set $S \subseteq V(G)$ we use $G[S]$ to denote the subgraph of $G$ induced by $S$ and $G - S$ to denote the graph $G[V(G)\setminus S]$. For $v \in V(G)$, we use $G-v$ as shorthand for $G-\{v\}$. We define $\bar{G}$ as the complement of $G$, that is, $V(\bar{G}) := V(G)$ and $E(\bar{G}) := \{\{u,v\} \mid \{u,v\} \notin E(G) \wedge u \neq v \wedge u,v\in V(G)\}$.

A set $S \subseteq V(G)$ is a \emph{vertex cover} of a graph $G$, if for each edge $e \in E(G)$, at least one of its endpoints is contained in $S$. We will use $\OPT(G)$ to denote the size of a minimum vertex cover of $G$.
A \emph{matching} in $G$ is a set $F \subseteq E(G)$ such that no two edges in $F$ share an endpoint. We say a vertex of $G$ is \emph{matched} by $F$ if it is incident to an edge in $F$. We use $\MM(G)$ to denote the size of a maximum matching in $G$.

\begin{definition}[{$\LP(G)$}]
Let $G$ be a graph. The linear program relaxation for \VC for $G$ is defined as
\[\LP(G) = \min\left\{\sum_{v \in V(G)} x_v \mid \forall \{u,v\} \in E(G) : x_u + x_v \geq 1 \wedge \forall v \in V(G) : 0 \leq x_v \leq 1\right\}.\]
A \emph{feasible solution} is a setting to the variables $x_v$ for all $v\in V(G)$ that satisfies the conditions of the linear program above. It is well known that for vertex cover, there is an optimal feasible solution such that $x_v \in \{0,\frac{1}{2},1\}$ for all $v\in V(G)$. We call a solution for which this holds a \emph{half-integral} solution.
\end{definition}
\begin{definition}\label{def:elimination-dist-to-C}
Let $\mathcal{C}$ be a graph class and let $G$ be a graph. We define the \emph{elimination distance to \C} as
\begin{align*}
     \ed{\C}{G} := \begin{cases}
                        0 & \text{, if }G \in \C \\
                        \min_{v \in V(G)}\ed{\C}{G-\{v\}}+1 & \text{, if } G \notin \C \text{ and $G$ is connected}\\
                        \max_{i \in [t]} \ed{\C}{G_i} & \hspace{-1.1cm}\text{, if $G$ consists of connected components $G_1,\ldots,G_t$. }
                    \end{cases}
\end{align*}
By this definition, the \emph{treedepth} of a graph $G$, denoted by $\td{G}$, corresponds to its elimination distance to the empty graph.

We can define an \emph{elimination forest} corresponding to the definition above. For a graph $G$, its elimination forest (w.r.t. \C) is a rooted forest $T$ (i.e. every connected component of $T$ has a root) of height $\ed{\C}{G}$ in which every vertex $x\in T$ has a bag $B_x\subseteq V(G)$ such that
\begin{itemize}
\item For every non-leaf $x \in V(T)$, $|B_x|=1$. In this case we may identify $x$ with the only element in $B_x$.
\item For every leaf $x \in V(T)$, $G[B_x] \in \C$ (we call these the \emph{base components} of $T$).
\item Every vertex $v \in G$ is contained in exactly one bag in $T$.
\item If $\{u,v\} \in E(G)$ and $u \in B_x$, $v \in B_y$ then either $x=y$, $y$ is an ancestor of $x$, or $x$ is an ancestor of $y$ in $T$.
\end{itemize}
Observe, if $G$ is connected then $T$ is a tree and has exactly one root.
\end{definition}

Let $\C$ be a graph class, let $G$ be a graph and let $X \subseteq V(G)$. We say $X$ is a \emph{\C-modulator} if $G-X \in \C$. We say $X$ is a \emph{$(\C,d)$-modulator} if $\ed{\C}{G-X} \leq d$. When considering \VC parameterized by the size of a \C-modulator or a $(\C,d)$-modulator, we will assume that this modulator is given on input. As such, inputs to the problem are triplets $(G,k,X)$ such that $X$ is a modulator in $G$ and the problem is to decide whether $G$ has a vertex cover of size $k$.

For a graph class \C,  let $\C + c$ be the graph class consisting of all graphs that have a \C-modulator of size at most $c$, such that $\C+c:=\{G \mid \exists X \subseteq V(G), |X|\leq c \colon G-X \in \C\}$ where $c$ is a positive integer.

\begin{definition}[Blocking set]
Let $G$ be a graph and let $Y \subseteq V(G)$ be a subset of its vertices. We say that $Y$ is a \emph{blocking set} in $G$ if there exists no vertex cover $S$ of $G$ such that $Y \subseteq S$ and $|S| = \OPT(G)$. In other words, there is no optimal vertex cover of $G$ that contains $Y$. A blocking set $Y$ is \emph{minimal} if no strict subset  of $Y$ is also a blocking set.
\end{definition}

Let $G$ be a graph, we use $\bsg{G}$ to denote the size of the largest minimal blocking set in $G$. For a graph class $\C$, let $\bsc{\C}:= \max_{G \in \C} \bsg{G}$, let $\bsc{\C} := \infty$ if the minimal blocking set size of graphs in this graph class can be arbitrarily large. Define $\bsd{\C}{d} := \max\{\bsg{G} \mid \ed{\C}{G} \leq d\}$.

\subparagraph*{Robustness} Throughout the paper we will assume all graph classes \C are \emph{robust}, meaning that they are closed under disjoint union and under removing connected components. Observe that this does not influence the blocking set size; it is easy to see that if a graph $G$ has connected components $G_1$ and $G_2$, then $\bsg{G} = \max\{\bsg{G_1},\bsg{G_2}\}$.

\subparagraph{Parameterized complexity}
A \emph{parameterized problem} \Q is a subset of $\Sigma^* \times \mathbb{N}$, where $\Sigma$ is some finite alphabet. While it is customary to denote the parameter value by $k$, in this paper we will generally use $\ell$ for the parameter value and $k$ for the (desired) solution size.

\begin{definition}\label{def:kernel}
Let \Q be a parameterized problem and let $f\colon \mathbb{N}\to \mathbb{N}$ be a computable function. A \emph{kernel for \Q of size $f(\ell)$} is an algorithm that, on input $(x,\ell) \in \Sigma^* \times \mathbb{N}$, takes time polynomial in $|x|+\ell$ and outputs an instance $(x',\ell')$ such that:
\begin{enumerate}
  \item $|x'|$ and $\ell'$ are bounded by $f(\ell)$, and
  \item $(x',\ell') \in \Q$ if and only if $(x,\ell) \in \Q$.
\end{enumerate}
The algorithm is a \emph{polynomial kernel} if $f(\ell)$ is a polynomial.
\end{definition}

We will often use that if there is a polynomial-time linear-parameter transformation from a parameterized problem \Q to a parameterized problem $\Q'$, then kernelization lower bounds established for \Q also hold for $\Q'$ \cite{BodlaenderJK14,BodlaenderTY11}.

We say a parameterized problem \Q is \emph{fixed-parameter tractable (FPT)}  if there is an algorithm solving it in running time $f(\ell)\cdot \text{poly}(n)$, where $n$ is the input size and $\ell$ is the parameter. It is well-known \cite{CyganFKLMPPS15} that \Q is FPT if and only if \Q admits a (not necessarily polynomial) kernel.

\subparagraph{Complexity-theoretic assumptions}
We say a language $L$ is in $\mathsf{RP}$ if there is a randomized polynomial-time Turing machine T such that (i) for all $x\notin L$, $T$ rejects $x$, and (ii) for all $x\in L$, $T$ accepts $x$ with probability at least $\frac{1}{2}$. By this definition, it follows straightforwardly that $\mathsf{P \subseteq RP\subseteq NP}$ and $\mathsf{RP \subseteq BPP}$~\cite{KO198239}.

In this paper we will obtain results under the assumption that $\mathsf{NP \neq RP}$. Clearly, if the two often believed conjectures $\mathsf{P = BPP}$ and $\mathsf{P \neq NP}$ hold, we must have that $\mathsf{P = RP}$ and $\mathsf{NP \neq RP}$. Furthermore, this assumption is weaker than the standard assumption \ncontainment under which kernelization lower bounds are obtained. Since it is known that $\mathsf{RP \subseteq P/poly}$ \cite[Theorem I]{Adleman78}, it follows that $\mathsf{RP = NP}$ would imply $\mathsf{NP \subseteq P/poly}$, implying \containment.

\section{Relation between minimal blocking sets and polynomial kernels}\label{section::relation}

In this section, we show relations between the size of minimal blocking sets in \C and kernelization bounds for \VC parameterized by a \C-modulator. We start by proving that when \C contains graphs with large minimal blocking set size, then this transfers to kernelization lower bounds for \VC parameterized by a \C-modulator. This shows that bounded minimal blocking set size is necessary for having a polynomial kernel. We then show, however, that bounded minimal blocking set size is not \emph{sufficient} for getting a polynomial kernel parameterized by a modulator to $\C$. Concretely, we exhibit a graph class \C with blocking set size $1$, for which \VC parameterized by a \C-modulator has no polynomial kernel unless $\RP=\NP$. We end with a positive result: under some additional conditions on \C, if \C has bounded minimal blocking set size this implies that we can efficiently reduce the number of connected components in $G-X$, if $G$ is a graph and $X$ is a \C-modulator.

\subsection{Polynomial kernel implies a bound on the minimal blocking set size}\label{section::kernelimpliesbmbss}
\newcommand{\vcopt}{\ensuremath{\mathrm{OPT}}\xspace}
\newcommand{\dVC}{$d$-\textsc{Vertex Cover}\xspace}
\newcommand{\dprimeVC}{$d$-\textsc{Vertex Cover}\xspace}

In this section we prove Theorem~\ref{theorem:lb:intro}, showing that if $\C$ is a graph class where minimal blocking sets can have size $d$, then this gives a kernelization lower bound for \VC when parameterized by the size of a \C-modulator. Thus, under the assumption that \ncontainment, the  theorem shows that having bounded blocking set size is necessary to obtain a polynomial kernel in the following sense. For a graph class \C closed under disjoint union, for which \VC parameterized by a modulator to \C admits a polynomial kernel of size $\Oh(k^d)$, it must hold that $\bsc{\C} \leq d$.

\begin{thm:theorem1}
\theoremone
\end{thm:theorem1}
\begin{proof}
For $d = 1$,  observe that a kernel of size $\Oh(\ell^{1-\varepsilon})$ can be ruled out by the following argument. Suppose such a kernel exists for $\varepsilon > 0$. Let $(G,k)$ be an input instance. Let $H$ be an arbitrary constant-size graph in $\C$. Let $G'$ be the disjoint union of $H$ and $G$, implying that $G'$ has a modulator to \C of size $|V(G)|$. Let $k' := k + \vcopt(H)$. It is easy to observe that $G'$ has a vertex cover of size $k'$ if and only if $G$ has a vertex cover of size $k$. However, using the hypothetical kernel we can solve $(G',k')$ in polynomial time, by repeatedly applying the kernelization algorithm until we obtain a constant-size instance. This would imply that $\mathsf{P = NP}$, implying \containment.

For $d\geq 2$, the lower bound is obtained by a linear-parameter transformation from \dprimeVC. An input to this problem consists of a hypergraph $G$ and integer $k$ where every hyperedge consists of exactly $d$ vertices. The problem is to decide whether $G$ has a vertex cover of size at most $k$. A vertex cover of a hypergraph is a set $X \subseteq V(G)$ such that for every edge $e \in E(G)$ we have $e \cap X \neq \emptyset$.

The lower bound will then follow from the fact that for $d\geq 2$, \dprimeVC parameterized by the number of vertices $n$ does not have a kernel of size $\Oh(n^{d-\varepsilon})$ for any $\varepsilon > 0$ unless \containment \cite[Theorem~2]{DellM14Satisfiability}.

Suppose we are given an instance $(G,k)$ for \dprimeVC with $V(G) = \{x_1,\ldots,x_n\}$ and $E(G) = \{e_1,\ldots,e_m\}$, we show how to construct an instance $(G',k')$ for \VC. Refer to Figure~\ref{fig:lpt} for a sketch of $G'$.

\begin{figure}
\centering
\includegraphics{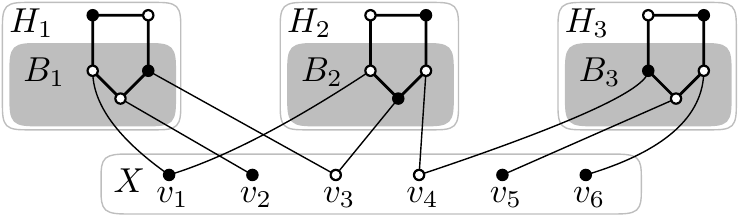}
\caption{The \VC instance $G'$ obtained in the proof of Theorem~\ref{theorem:lb:intro}, corresponding to instance $(G,2)$ with $V(G) = \{x_i \mid i \in [6]\}$, and $E(G)= \{\{x_1,x_2,x_3\},\{x_1,x_3,x_4\},\{x_4,x_5,x_6\}\}$. Since $G$ has a vertex cover of size $2$, $G'$ has a vertex cover of size $2 + 3\cdot3$,  indicated in white.}
\label{fig:lpt}
\end{figure}

Start by adding a vertex $v_i$ for all $i \in [n]$. Let $X := \{v_i\mid i\in[n]\}$, the set $X$ will be a modulator to $\mathcal{C}$ in $G$.
Let $H \in \C$ be a graph with a minimal blocking set of size $d$. For all $j \in [m]$, create a new copy of $H$ called $H_j$ and add it to $G'$. Choose a minimal blocking set $B_j$ of size $d$ in $H_j$ and enumerate the vertices as $B_j := \{b^j_1,\ldots,b^j_{d}\}$.

We now connect the vertices of graphs $H_j$ to vertices in $X$, depending on the vertices contained in edge $e_j$.  For each $j \in [m]$, for each $q \in [d]$, if the $q$'th vertex in edge $e_j$ equals $x_i$, connect vertex $b^j_q$ to vertex $v_i$. This concludes the construction of $G'$.

Observe that $X$ is a modulator to $\C$ in $G'$ since $G' - X$ consists of disjoint copies of $H \in \C$, and \C is closed under disjoint union.
Furthermore, $|X| = n = \Oh(n)$ which is appropriately bounded for a linear-parameter transformation. Let $k' := m\cdot \vcopt(H) + k$, we show that $G'$ has a vertex cover of size $k'$, if and only if $G$ has a vertex cover of size $k$.

$(\Rightarrow)$ Suppose $G$ has a vertex cover $S$ of size $k$, we show how to construct a vertex cover $S'$ of size $k'$ in $G'$. For $i \in [n]$, if $x_i \in S$ then let $v_i \in S'$. This completely defines $X\cap S'$, we show how to extend $S'$ to a vertex cover of the entire graph, using at most $\vcopt(H)$ vertices from each copy of $H$.

For all $j \in [m]$, let $B_j':=\{b \in B_j \mid \{b,x\} \in E(G') \wedge x \in X \setminus S'\}$. Observe that $B_j' \subsetneq B_j$, since there is at least one vertex $b \in B_j$ such that its  neighbor in $X$ is contained in $S'$, as $S$ was a vertex cover of $G$.
Add a minimum vertex cover of $H_j$ that contains $B_j'$ to $S'$. Since $B_j'\subsetneq B_j$ and $B_j$ is a minimal blocking set for $H_j$, it follows that this vertex cover has size $\vcopt(H)$.
Hereby, $S'$ has size $k'$ and it is easy to verify that $S'$ is indeed a vertex cover of~$G'$.

$(\Leftarrow)$ Suppose $G'$ has a vertex cover $S'$ of size $k'=k+m\cdot \vcopt(H)$. We start by showing that we can modify $S'$ such that $|S' \cap H_j|=\vcopt(H)$ for all $j\in[m]$. Suppose for some $j\in[m]$, $|S' \cap V(H_j)|>\vcopt(H)$. Choose an arbitrary $b \in B_j$ and add $N(b) \cap X$ to $S'$, note that $|N(b) \cap X|=1$. Replace $S'\cap V(H_j)$ by a vertex cover of $H_j$ of size $\vcopt(H)$ that contains all vertices in $B_j\setminus \{b\}$, observe that such a vertex cover exists since $B_j$ is a minimal blocking set. As such, we from now on assume that $|S' \cap V(H_j)|=\vcopt(H)$ for all $j \in [m]$.

Define $S:= \{x_i \mid v_i \in S'\cap X, i\in[n]\}$ and observe that $|S| = k$. We show that $S$ is a vertex cover of $G$. Suppose there is an edge $e_j \in E(G)$ such that $e_j \cap S = \emptyset$. But then, for all $b \in B_j$ we have that there exists a vertex $x \in X$ such that $\{b,x\} \in E(G')$ and $x \notin S'$. Since $S'$ is a vertex cover of $G'$ that implies $S' \cap V(H_j)$ is a vertex cover of $H_j$ containing all vertices from $B_j$. This however contradicts the fact that $|S'\cap V(H_j)| = \vcopt(H)$ and $B_j$ is a blocking set.
\end{proof}

The above theorem shows that the existence of large minimal blocking sets allows us to prove kernelization lower bounds for \VC. It can be seen that many of the existing lower bound results for \VC when parameterized by a modulator to~\C for some graph class \C, relied on this same idea. Jansen \cite[Theorem~5.3]{Jansen13Thesis} showed that \VC parameterized by a modulator to outerplanar graphs has no polynomial kernel, by a polynomial parameter transformation starting from \textsc{$d$-cnf-sat}. The proof relies on the construction of an outerplanar clause gadget with blocking set size $d$, for any $d$. The same construction was used to obtain a lower bound for \VC parameterized by a modulator to constant treedepth \cite[Theorem~2]{JansenP18Polynomial}. Furthermore, we the mention the proof that \VC parameterized by the size of a modulator to a mock forest (a graph where no two cycles share a vertex) is unlikely to have a polynomial kernel \cite[Theorem~2]{FominS16}. Again, here it is used that a mock forest has unbounded minimal blocking set size.

When choosing \C as the graph class consisting of all cluster graphs in which each clique has size at most $d$, a lower bound was obtained by Majumdar et al.~\cite{MajumdarRS18}. They show the problem is unlikely to have a kernel of size $\Oh(n^{d-\varepsilon})$, again  relying on the fact that a size-$d$ clique has minimal blocking set size $d$.

As such, Theorem~\ref{theorem:lb:intro} generalizes most existing lower bounds for \VC when parameterized by the size of a modulator to \C for some graph class \C. We however mention that \VC parameterized by the size of a modulator to a clique has also been shown to not admit a polynomial kernel \cite{BodlaenderJK14}, unless \containment. This result does not fit our framework. While cliques have unbounded blocking set size, the class of cliques is not closed under disjoint union. This introduces additional difficulties when constructing the type of transformation we give above. Indeed, the lower bound for \VC parameterized by the size of a modulator to a clique is obtained by a cross-composition.

\subsection{Bounded minimal blocking set size is not sufficient}\label{section::bmbssisnotsufficient}
\newcommand{\USAT}{\textsc{Unique-SAT}\xspace}
\newcommand{\UTSAT}{\textsc{Unique-$3$-SAT}\xspace}
Now that it is clear that, proving that a graph class has bounded blocking set size is essential towards obtaining a polynomial kernel for \VC parameterized by the size of a modulator to this graph class, one may wonder whether this condition is  also sufficient. It turns out that it is \emph{not}. Even worse, there exists a graph class \C for which all minimal blocking sets have size $1$, for which \VC is unlikely to be solvable in polynomial time. As such, \VC parameterized by the size of a modulator to \C is unlikely to be FPT, as \VC is not likely to be solvable in polynomial time when the parameter value is zero. Since a problem that is not FPT does not admit any kernel at all, this proves our result. More precisely, we obtain the following theorem. (The graph class \C obtained in the proof could be made robust by taking its closure under disjoint union and deletion of components without breaking the fact that a polynomial-time algorithm for \VC is unlikely.)

\begin{thm:theorem2}
\theoremtwo
\end{thm:theorem2}
\begin{proof}
It is known that the \USAT problem cannot be solved in polynomial-time unless $\mathsf{NP = RP}$ \cite[Corollary 1.2]{DBLP:journals/tcs/ValiantV86}. An input to \USAT is a CNF-formula $\mathcal{F}$ that has either exactly one satisfying solution or is unsatisfiable. The problem is to decide whether $\mathcal{F}$ is satisfiable. It can be shown that the same result holds for \UTSAT \cite[Example 26.7]{DBLP:books/daglib/0071219}, where the input formula is further restricted to be in $3$-CNF.

We show that the following polynomial-time reduction from \UTSAT to \VC exists.
\begin{claim}
There is a polynomial-time reduction from \UTSAT to \VC, that given a formula $\mathcal{F}$, outputs an instance $(G,k)$ for \VC such that:
\begin{itemize}
\item If $\mathcal{F}$ has exactly one satisfying assignment, then $G$ has a unique minimum vertex cover of size $k$.
\item If $\mathcal{F}$ is unsatisfiable, then $G$ has a unique minimum vertex cover of size $k + 1$.
\end{itemize}
\end{claim}
\begin{claimproof}
Instead of showing the above result for \VC, we will for simplicity reduce from \UTSAT to \textsc{Clique}. We give a reduction such that when $\mathcal{F}$ has a unique satisfying assignment, $G$ has a unique clique of size $\ell$, and otherwise $G$ has a unique clique of size $\ell-1$. Since $S \subseteq V(G)$ is a clique in $G$ if and only if $V(G)\setminus S$ is a vertex cover in $\bar{G}$, the desired reduction for \VC follows from taking the complement of $G$ and letting $k := |V(G)|-\ell$.

Let a formula $\mathcal{F}$ be given. There is a parsimonious reduction from \UTSAT to \textsc{Clique} (see for example \cite[Example 26.8]{DBLP:books/daglib/0071219}), meaning that there is a polynomial-time many-one reduction from \UTSAT to \textsc{Clique} such that the number of size-$\ell$ cliques in $G$ corresponds to the number of satisfying assignments of $\mathcal{F}$.  Use this reduction to obtain an instance $(G,\ell)$ for \textsc{Clique}. Clearly, $G$ has the property that if $\mathcal{F}$ has a unique satisfying assignment, then~$G$ has a unique maximum clique. However, we also want to guarantee that~$G$ has a unique (albeit smaller) maximum clique when $\mathcal{F}$ is unsatisfiable.

 We now show how to obtain an instance $G'$ satisfying both these requirements. Initialize~$G'$ as two copies of graph $G$, say $G_1$ and $G_2$. Label the vertices of $G_1$ and $G_2$ as $v_1^1,\ldots,v_n^1$ and $v_1^2,\ldots,v_n^2$ such that $G_1$ and $G_2$ are isomorphic by the function mapping $v_i^1$ to $v_i^2$ for all $i \in [n]$.

For all $i \in [n]$, add the edge $\{v_i^1,v_i^2\}$ to $G'$. Furthermore,
 add the edges  $\{\{v_i^1,v_j^2\} \mid \{v_i^1,v_j^1\}\in E(G_1)\}$ to $G'$. In this way, vertices $v_i^1$ and $v_i^2$ have the same closed neighborhood in $G'$, for any $i \in [n]$.

 Conclude the construction of $G'$ by adding a new set of vertices $Z$ of size $2\ell-1$ to the graph, and letting $Z$ be a clique. We show that the conditions of the lemma statement hold for $(G',\ell')$ where we let $\ell' := 2\ell$.

 Suppose $\mathcal{F}$ has exactly one satisfying assignment. By this definition, there is a size-$\ell$ clique in $G_1$, let this clique be $K_1$. Define $K_2 := \{v_i^2\mid v_i^1 \in K_1, i \in [n]\}$. Clearly, $K_2$ is a clique in $G_2$, we show that $K_1 \cup K_2$ is a clique in $G'$. Let $u,v \in K_1\cup K_2$ with $u \neq v$. If $u,v \in K_1$ or $u,v\in K_2$ it follows $\{u,v\} \in E(G')$, so assume without loss of generality that $u = v_i^1$ and $v = v_j^2$ for $i,j\in[n]$. If $i=j$, $\{v_i^1,v_i^2\} \in E(G')$ by definition. If $i \neq j$, then $v_j^1 \in K_1$. As such, $\{v_i^1,v_j^1\} \in E(G_1)$ since $K_1$ is a clique, and thereby we again have that $\{v_i^1,v_j^2\} \in E(G')$, as desired.
 Observe that this clique is unique. No size-$\ell'$ clique in $G'$ can contain vertices from $Z$. Furthermore, $G_1$ and $G_2$ have no cliques of size larger than $\ell$. As such, any size-$\ell'$ clique $K$ in $G'$ has the property that $K \cap V(G_1) = \ell$ and $K \cap V(G_2) = \ell$. If $\mathcal{F}$ has exactly one satisfying assignment, then size-$\ell$ cliques are unique in $G_1$ and $G_2$.

 Suppose $\mathcal{F}$ is unsatisfiable. Thereby, $G_1$ and $G_2$ do not have a clique of size $\ell$. It is easy to see that the subgraph of $G'$ induced by the vertices of $G_1$ and $G_2$ thereby has no clique of size $\ell'-1$. Thereby, $Z$ is a unique size-$(\ell'-1)$ clique in $G'$, concluding the proof.
\end{claimproof}

To conclude the proof, let $\C$ be the graph class consisting of all graphs that are obtained via the reduction given in the claim above, when starting from a formula $\mathcal{F}$ that has zero or one satisfying assignments. As such, solving \VC on \C in polynomial time corresponds to solving \UTSAT in polynomial time, implying $\mathsf{NP = RP}$. Since any graph in \C has exactly one minimum vertex cover, we obtain that indeed $\bsc{\C} = 1$, as any vertex that is not part of the minimum vertex cover forms a (minimal) blocking set. \end{proof}

Observe that graphs in the graph class \C constructed in the proof of Theorem~\ref{theorem:intro:bounded-bs-not-enough} are always connected, since they are the complement of a disconnected graph. As such, \C is closed under removing connected components. However, \C is not robust because it is not closed under disjoint union. We can however define $\C'$ to contain all graphs for which all connected components lie in \C. Observe that $\C'$ is robust, but that $\bsc{\C'} = 1$ and \VC is not solvable in polynomial time on $\C'\supseteq \C$ unless $\mathsf{RP = NP}$.
\subsection{Reducing the number of components outside the modulator}
\label{subsec:reducing-num-components}
As mentioned in the previous subsections, bounded blocking set size is necessary to obtain polynomial kernels for \VC. Many papers that give polynomial kernels for \VC parameterized by the size of a \C-modulator showed that their graph class \C has bounded blocking set size, see for example \cite{BougeretS17,FominS16,HolsK17,JansenB13,MajumdarRS18}. Some of them used the blocking set size of class \C to bound the number of connected components. More precisely, given an instance $(G,k,X)$ of \VC with $G-X \in \C$ they showed that one can reduce the number of connected components of $G-X$ to $\Oh(|X|^{\bsc{\C}+1})$. We will show that one can reduce the number of connected components of $G-X$ to $|X|^{\bsc{\C}}$, as a first step towards proving Theorem~\ref{intro:thm:reduce-cc}. Here we assume that the class \C is robust in order to guarantee that deletion of connected components of $G-X$ again results in a graph of \C. In Section~\ref{subsubsec:poly-time-rrule-application} we discuss suitable conditions so that this component reduction can be done efficiently.

Let $(G,k,X)$ be an instance of \VC parameterized by the size of a \C-modulator.
First, we define the set $\X=\{Z \subseteq X \mid Z \text{ is an independent set in } G \text{ and } 1\leq|Z|\leq \bsc{\C}\}$ as the collection of \emph{chunks} of $X$. The intuition of defining the set $\X$ of chunks is to find sets in the modulator $X$ for which at least one vertex must be contained in any optimum vertex cover of $G$. The concept of chunks was first introduced by Jansen and Bodlaender \cite{JansenB13}.

To reduce the number of connected components of $G-X$, we use the following result due to Hopcroft and Karp \cite{HopcroftK73} which computes a certain crown-like structure in a bipartite graph. The second part of the theorem is not standard (but well known).

\begin{theorem}[\cite{HopcroftK73}] \label{theorem::matching}
    Let $G$ be an undirected bipartite graph with bipartition $V_1$ and $V_2$, on $n$ vertices and $m$ edges. Then we can find a maximum matching of $G$ in time $\Oh(m \sqrt{n})$. Furthermore, in time $\Oh(m \sqrt{n})$ we can find either a maximum matching that saturates $V_1$ or a set $Z \subseteq V_1$ such that $|N_G(Z)|<|Z|$ and such that there exists a maximum matching $M$ of $G-N_G[Z]$ that saturates $V_1 \setminus Z$.
\end{theorem}

We construct a bipartite graph $G_B$ to which we will apply Theorem \ref{theorem::matching} to find a set of connected components in $G-X$ that can be safely removed from $G$.
We denote the set of connected components in $G-X$ by $\F$. The two parts of the bipartite graph $G_B$ are the set $\X$ of chunks and the set $\F$ of connected components in $G-X$. More precisely, for every chunk $Z \in \X$ and for every connected component $H \in \F$ we add a vertex to the bipartite graph. To simplify notation we denote the vertex of $G_B$ that corresponds to a connected component $H \in \F$ resp.\ a chunk $Z \in \X$ by $H$ resp.\ $Z$. We add an edge between a vertex $H \in \F$ and a vertex $Z \in \X$ when $N_G(Z) \cap V(H)$ is a blocking set in $H$, i.e., when $\OPT(H-N_G(Z))+|N_G(Z) \cap V(H)| > \OPT(H)$.

It follows from Theorem \ref{theorem::matching} that there exists either a maximum matching $M$ of $G_B$ that saturates $\X$ or a set $\X' \subseteq \X$ such that $|N_{G_B}(\X')| <|\X'|$ and such that there exists a maximum matching $M$ of $G_B-N_{G_B}[\X']$ that saturates $\hat{\X}=\X \setminus \X'$.
If there exists a maximum matching $M$ of $G_B$ that saturates $\X$ then let $\X'=\emptyset$ and let $\hat{\X}=\X$. Let $\F_D = \F \setminus (N_{G_B}(\X') \cup V(M))$ be the set of connected components in $\F$ that are neither in the neighborhood of $\X'$ nor endpoint of a matching edge of $M$.

\begin{redrule} \label{rule::delete_cc}
    Delete all connected components in $\F_D$ from $G$ and decrease the size of $k$ by $\OPT(\F_D)$ the size of an optimum vertex cover in $\F_D$.
\end{redrule}

Observe that Reduction Rule \ref{rule::delete_cc} deletes also all connected components $H \in \F$ which have the property that for all sets $Z \in \X$ it holds that $N(Z) \cap V(H)$ is not a blocking set of $H$ because these connected components correspond to isolated vertices in the bipartite graph $G_B$.
Before we show the correctness of Reduction Rule \ref{rule::delete_cc} we show the following lemma which guarantees us the existence of certain optimum vertex covers of $G$.

\begin{lemma} \label{lemma::reduce_cc}
    There exists an optimum vertex cover $S$ of $G$ with $S \cap Z \neq \emptyset$ for all $Z \in \hat{\X}$.
\end{lemma}

\begin{proof}
    Let $S$ be an optimum vertex cover of $G$. If $S \cap Z \neq \emptyset$ for all $Z \in \hat{\X}$ then we are done, since $S$ fulfills the requirements of the lemma.
    Thus, assume that there exists at least one set $Z \in \hat{\X}$ such that $S \cap Z = \emptyset$.
    Let $\widetilde{\X}=\{Z \in \hat{\X} \mid S \cap Z = \emptyset\}$ be the set of sets in $\hat{\X}$ that have an empty intersection with the optimum vertex cover $S$, and let $\widetilde{\F} = \{H \in \F \mid \exists Z \in \widetilde{\X} \colon \{Z,H\} \in M\}$ be the set of connected component in $\F$ that are matched to a vertex in $\widetilde{\X}$ via the matching $M$ that saturates $\hat{\X}$; thus, $|\widetilde{\X}|=|\widetilde{\F}|$.
    \begin{claim} \label{claim::reduce_cc_not_opt}
        It holds for all connected components $H \in \widetilde{\F}$ that $|V(H) \cap S| > \OPT(H)$.
    \end{claim}
    \begin{claimproof}
        For every connected component $H \in \widetilde{\F}$ let $Z_H \in \widetilde{\X}$ be the chunk with $\{H,Z_H\} \in M$. Thus, the set $Y_H = N_G(Z_H) \cap V(H)$ is a blocking set in $H$. Since $Z_H \cap S = \emptyset$ it holds that the vertex cover $S\cap V(H)$ of $H$ must contain the blocking set $Y_H$; hence $|V(H) \cap S| > \OPT(H)$.
    \end{claimproof}
    Now, we construct an optimum vertex cover $S'$ of $G$ which fulfills the properties of the lemma. First, we add from each chunk $Z \in \widetilde{\X}$ one arbitrary vertex to the set $S$. We denote the resulting set by $\hat{S}$. It holds that $|\hat{S}| \leq |S| + |\widetilde{\X}|$.
    \begin{claim} \label{claim::reduce_cc_sol_size}
        For every connected component $H \in \widetilde{\F}$ there exists an optimum vertex cover $S_H$ of $H$ which contains the set $Y_H = N_G(X\setminus \hat{S}) \cap V(H)$.
    \end{claim}
    \begin{claimproof}
        Assume for contradiction that the claim does not hold. This implies that the set $Y_H$ is a blocking set of $H$. Let $Y'_H \subseteq Y_H$ be a minimal blocking set of $H$. Since $H$ is a connected component of a graph of the graph class $\C$ it holds that $|Y'_H| \leq \bsc{\C}$.
        For every vertex $y \in Y'_H$ choose an arbitrary vertex $x \in N_G(y) \cap (X\setminus \hat{S})$ and denote the resulting set by $Z_H$. The set $Z_H$ has size at most $\bsc{\C}$ and is an independent set because it is a subset of the independent set $X \setminus \hat{S}$. Therefore, $Z_H$ is a chunk in $\X$. By construction, the neighborhood of $Z_H$ in $H$ is a superset of $Y'_H$ and thus a blocking set in $H$. It follows that $Z_H$ is a chunk in $\hat{\X}$ because $\{Z_H,H\} \in E(G_B)$, and because $H$ is a vertex in $G_B-N_{G_B}[\X']$. But, by construction, every chunk $Z$ in $\hat{\X}$ has a nonempty intersection with $\hat{S}$. This is a contradiction to the assumption that $Y_H$ is a blocking set of $H$ and proves the claim.
    \end{claimproof}
    In a next step, we replace for every connected component $H \in \widetilde{\F}$ the set $S\cap V(H)$ in $\hat{S}$ by an optimum vertex cover $S_H$ in $H$ that contains the set $N_G(X\setminus \hat{S})$. The existence of such an optimum vertex cover $S_H$ follows from Claim \ref{claim::reduce_cc_sol_size}. We denote the resulting set by $S'$. It follows from Claim \ref{claim::reduce_cc_not_opt} that the set $S'$ has size at most $|S|$ because we replace for each connected component $H \in \widetilde{\F}$ the non-optimum vertex cover $S \cap V(H)$ by the optimum vertex cover $S_H$, more precisely, $|S'|\leq |\hat{S}| - |\widetilde{\F}| \leq |S| + |\widetilde{\X}| - |\widetilde{\F}| = |S|$.

    It remains to show that $S'$ is a vertex cover of $G$. We only add vertices of the modulator $X$ to the vertex cover and we change the vertex cover of the connected components in $\widetilde{\F}$. Thus, any edge that is possibly not covered by $S'$ must either be contained in one of the connected components of $\widetilde{\F}$ or between such a connected component and the modulator $X$. Both is not possible, because we add for each connected component $H \in \widetilde{\F}$ a vertex cover $S_H$ to the set $\hat{S}$ which contains the neighborhood of all vertices of $X$ that are not in the vertex cover $\hat{S}$. This concludes the proof.
\end{proof}

Now, we show the correctness of Reduction Rule \ref{rule::delete_cc} using Lemma \ref{lemma::reduce_cc}. Let $(\widetilde{G},\widetilde{k},X)$ be the reduced instance, i.e., $\widetilde{G}=G-\F_D$ and $\widetilde{k}=k-\OPT(\F_D)$. Obviously, if $(G,k,X)$ is a yes-instance then $(\widetilde{G},\widetilde{k},X)$ is a yes-instance. For the other direction, assume that $(\widetilde{G},\widetilde{k},X)$ is a yes-instance. Observe that $M$ is also a matching in $\widetilde{G}_B$ that saturates $\hat{\X}$ because we delete no connected component that is an endpoint of a matching edge. Furthermore, it holds that either $\hat{\X}=\X$ or $|N_{\widetilde{G}_B}(\X')|<|\X'|$ because we delete no connected component that corresponds to a vertex in $|N_{\widetilde{G}_B}(\X')|$. Thus, it follows from Lemma \ref{lemma::reduce_cc} that there exists an optimum vertex cover $\widetilde{S}$ of $\widetilde{G}$ with $\widetilde{S} \cap Z \neq \emptyset$ for all sets $Z \in \hat{\X}$.
Note that every connected component $H \in \F_D$ is only adjacent to vertices in $\hat{\X}$ in $G_B$. Since every set $Z \in \hat{\X}$ has a non-empty intersection with the set $\widetilde{S}$, it holds that there exists an optimum vertex cover $S_H$ of $H$ which contains the set $N_G(X\setminus \widetilde{S}) \cap V(H)$ (similar to Claim \ref{claim::reduce_cc_sol_size}).
Let $S$ be the set that results from adding for each connected component $H \in \F_D$ the optimum vertex cover $S_H$ to the set $\widetilde{S}$. By construction, it holds that $S$ is a vertex cover of $G$ of size $|\hat{S}|+\OPT(\F_D) \leq \hat{k} + \OPT(\F_D) =k$. This proves that $(G,k,X)$ is also a yes-instance. Overall, we showed that Reduction Rule \ref{rule::delete_cc} is safe.

\begin{theorem} \label{theorem::bound_numb_cc}
    Let $(G,k,X)$ be an instance of \VC parameterized by the size of a $\C$-modulator that is reduced under Reduction Rule \ref{rule::delete_cc}. The graph $G-X$ has at most $|X|^{\bsc{\C}}$ connected components.
\end{theorem}

\begin{proof}
    Since instance $(G,k,X)$ is reduced under Reduction Rule \ref{rule::delete_cc} every connected component of $G-X$ is either in $N_{G_B}(\X')$ or an endpoint of a matching edge in $M$. Recall that $M$ is a matching in $G-N_{G_B}[\X']$ that saturates all vertices in $\hat{\X}$. Thus, the number of connected components in $G-X$ is bounded by $|N_{G_B}(\X')| + |\hat{\X}| \leq |\X'|+ |\hat{\X}| = |\X|$. The set $\X$ of chunks contains at most $\sum_{i=1}^{\bsc{\C}} \binom{|X|}{i} \leq |X|^{\bsc{\C}}$ many sets; hence $G-X$ has at most $|X|^{\bsc{\C}}$ many connected components.
\end{proof}

\subsubsection{Applying Reduction Rule \ref{rule::delete_cc} in polynomial time}
\label{subsubsec:poly-time-rrule-application}

To use Theorem~\ref{theorem::bound_numb_cc} to prove that we can efficiently reduce the number of connected components in $G-X$ when $X$ is a \C-modulator, we need to show that, under certain assumptions, Reduction Rule~\ref{rule::delete_cc} can be applied in polynomial time. We start by providing two sufficient conditions in the next lemma.

\begin{lemma} \label{lemma::red_rule_poly_time}
    If \bsc{\C} is bounded, if \VC is solvable in polynomial time on graphs of class $\C$ and if we can verify in polynomial time whether a given set $Y$ is a blocking set in a graph of class \C then we can apply Reduction Rule \ref{rule::delete_cc} in polynomial time.
\end{lemma}
\begin{proof}
    To apply Reduction Rule \ref{rule::delete_cc} to an instance $(G,k,X)$ of \VC parameterized by the size of a $\C$-modulator, we have to construct the bipartite graph $G_B$. Therefore, we have to figure out for each set $Z \in \X$ and for every connected component $H$ of $G-X$ whether $N_G(Z) \cap V(H)$ is a blocking set in $H$. We can do this in polynomial time when $\bsc{\C}$ is constant and when we can verify in polynomial time whether $N_G(Z) \cap V(H)$ is a blocking set in $H$. Thus, under the assumptions of the lemma we can construct $G_B$ in polynomial time. Using Theorem \ref{theorem::matching} we can compute $\X'$, $\hat{X}$ and $\F_D$ in polynomial time. Hence, Reduction Rule \ref{rule::delete_cc} is applicable in polynomial time.
\end{proof}

We continue by providing two cases that satisfy the preconditions for the lemma above, such that Reduction Rule~\ref{rule::delete_cc} can be applied in polynomial time on these graph classes.

First of all, we consider the case that graph class \C is hereditary. In this case, being solvable in polynomial time on the class \C is sufficient to also be able to verify whether a given subset of the vertices is a blocking set, thus allowing us to apply Reduction Rule~\ref{rule::delete_cc} in polynomial time.
 As mentioned in Subsection \ref{subsec:reducing-num-components} we also need that \bsc{\C} is bounded.
Overall, we assume that \C is a hereditary graph class on which \VC is polynomial-time solvable and where \bsc{\C} is bounded.

\begin{lemma}\label{lem:rrule-1-in-poly-time-hereditary}
Let \C be any hereditary graph class on which \VC can be solved in polynomial time and where \bsc{\C} is bounded. Then Reduction Rule~\ref{rule::delete_cc} can be applied in polynomial time.
\end{lemma}
\begin{proof}
    We only have to prove that graph class \C fulfills the requirements of Lemma \ref{lemma::red_rule_poly_time}. Since we assumed that \bsc{\C} is bounded and that \VC is polynomial-time solvable on graphs of graph class \C, it remains to show that we can verify whether a given vertex set $Y$ is a blocking set of a graph $G$ of graph class \C. For this purpose, we compute the size of an optimum vertex cover of $G$ and an optimum vertex cover of $G-Y$. Both is possible in polynomial time because $G$ and $G-Y$ are graphs of graph class \C due to the fact that \C is hereditary. Now, if $\OPT(G)=\OPT(G-Y)+|Y|$ then $Y$ is not a blocking set of $G$, and if $\OPT(G)<\OPT(G-Y)+|Y|$ then $Y$ is a blocking set of $G$. Thus, \C fulfills the requirements of Lemma \ref{lemma::red_rule_poly_time} which proves that we can apply Reduction Rule \ref{rule::delete_cc} in polynomial time.
\end{proof}

Theorem~\ref{intro:thm:reduce-cc} (restated below) now follows directly from Theorem~\ref{theorem::bound_numb_cc} and Lemmas~\ref{lemma::red_rule_poly_time} and~\ref{lem:rrule-1-in-poly-time-hereditary}.
\begin{thm:theorem3}
\theoremthree
\end{thm:theorem3}

We can actually further generalize Theorem~\ref{intro:thm:reduce-cc} to some non-hereditary graph classes. However, we have more problems to show that Lemma \ref{lemma::red_rule_poly_time} holds for non-hereditary graph classes, because after deleting vertices from a graph $G$ that is contained in a non-hereditary graph class \C we do not know whether the resulting graph still belongs to the graph class \C. As such we need the additional assumption that \VC is also polynomial-time solvable on graph class $\C+1$.

This additional assumption is not unreasonable, when our goal is to obtain a kernelization algorithm for \VC. In fact, in order to obtain any kernel for \VC parameterized by the size of a modulator to \C it is necessary to assume that the problem is FPT. From this, it immediately follows that we can solve  \VC  in polynomial time on $\C + 1$.

\begin{lemma} \label{lemma::non_hereditary_verify_bs}
    Let $\C$ be a (possibly non-hereditary) graph class. If we can solve \VC in polynomial time on graphs of graph class \C and $\C+1$ then we can verify in polynomial time whether a set $Y \subseteq V(G)$ is a blocking set of $G$.
\end{lemma}
\begin{proof}
    Let $G$ be a graph of graph class \C and let $Y \subseteq V(G)$ be any vertex set. We construct the graph $\hat{G}$ by adding one vertex $v_Y$ to $G$ and by connecting vertex $v_Y$ to every vertex in $Y$. If $Y$ is not a minimal blocking set of $G$ then $\OPT(\hat{G})=\OPT(G)$ because there exists an optimum vertex cover of $G$ that contains $Y$, which is also a vertex cover of $\hat{G}$ as it contains all neighbors of $v_Y$, and because every vertex cover of $\hat{G}$ restricted to $V(G)$ is a vertex cover of $G$. Now, assume that $Y$ is a blocking set of $G$. This implies that $\OPT(G)<\OPT(G-Y)+|Y|$. We will show that $\OPT(\hat{G})>\OPT(G)$. Let $\hat{S}$ be an optimum vertex cover of $\hat{G}$. If $v_Y \in \hat{S}$ then $\hat{S} \setminus \{v_Y\}$ is a vertex cover of $G$; thus $\OPT(G) \leq \OPT(\hat{G})-1$. If $v_Y \notin \hat{S}$ then we know that $Y \subseteq \hat{S}$. This implies that $|\hat{S}| \geq |Y| + \OPT(G-Y) > \OPT(G)$. Thus, we showed that $\OPT(G)<\OPT(\hat{G})$ when $Y$ is a blocking set of $G$.

    Overall, to verify whether $Y$ is a blocking set of $G$ we have to compare the size of an optimum vertex cover in $G$ and $\hat{G}$. Since $G$ is a graph of graph class \C and $\hat{G}$ is a graph of graph class $\C+1$ we can compute these optimum vertex covers in polynomial time. This concludes the proof.
\end{proof}

\begin{theorem}
    Let \C be any robust graph class with minimal blocking set size $d$ on which \VC can be solved in polynomial time. Furthermore, assume that \VC can be solved in polynomial time on graphs of graph class $\C+1$. There is an efficient algorithm that given $(G,k,X)$ such that $G-X\in \C$ returns an equivalent instance $(G',k',X)$ such that $G'-X\in\C$ has at most $\Oh(|X|^d)$ connected components.
\end{theorem}
\begin{proof}
    This result follows directly from Theorem~\ref{theorem::bound_numb_cc}, Lemma \ref{lemma::red_rule_poly_time} and Lemma \ref{lemma::non_hereditary_verify_bs}.
\end{proof}

\section{Minimal blocking sets in graphs of bounded elimination distance}
\label{section:bounded-blocking-sets}
As seen in the previous section, minimal blocking sets play an important role for \VC kernelization. In this section we try to combine different structural parameters by considering the minimal blocking set size of graphs that have elimination distance $d$ to some graph class \C that has bounded minimal blocking set size.
First of all, we show some basic facts about minimal blocking sets, as well as some useful connections between minimal blocking sets of a graph $G$ and certain subgraphs of $G$.

\begin{proposition} \label{proposition::bs_basics}
    Let $G$ be a graph and let $Y$ be a minimal blocking set of $G$.
    \begin{enumerate}[(i)]
        \item The set $Y$ contains no vertex that is contained in every optimum vertex cover of $G$. \label{enum::bs_basics_all}
        \item If $Y$ contains a vertex that is not contained in any optimum vertex cover of $G$ then $|Y|=1$. More precisely, the minimal blocking set $Y$ contains exactly one of these vertices. \label{enum::bs_basics_none}
        \item The set $Y$ is contained in only one connected component of $G$. \label{enum::bs_basics_one_cc}
        \item $\OPT(G-Y) +|Y| = \OPT(G) + 1$ \label{enum::bs_basics_sol_size}
    \end{enumerate}
\end{proposition}

\begin{proof}
    To prove item (\ref{enum::bs_basics_all}) we assume for contradiction that there exists a vertex $y \in Y$ that is contained in every optimum vertex cover of $G$. Since $Y$ is a minimal blocking set of $G$ it holds that there exists an optimum vertex cover $X$ of $G$ that contains the set $Y \setminus \{y\}$. But, the vertex $y$ is contained in every optimum vertex cover of $G$, hence $Y \subseteq X$. This contradicts the assumption that $Y$ is a minimal blocking set of $G$ and proves that $Y$ contains no vertex that is contained in every optimum vertex cover of $G$.

    Let $v$ be a vertex that is not contained in any optimum vertex over of $G$. This implies that $\{v\}$ is a blocking set of $G$ (by definition). Thus, item (\ref{enum::bs_basics_none}) holds.

    Next, we prove item (\ref{enum::bs_basics_one_cc}). Assume for contradiction that the minimal blocking set $Y$ is contained in at least two connected components of $G$. Let $G'$ be a connected component of $G$ that contains at least one vertex of $Y$, and let $Y'$ be the set of vertices in $Y$ that are also contained in $G'$, i.e., $Y' = Y \cap V(G') \neq \emptyset$. Let $y \in Y'$ be an arbitrary vertex. Since $Y$ is a minimal blocking set of $G$ it holds that there exists an optimum vertex cover $X$ of $G$ with $Y \setminus \{y\} \subseteq X$. Thus, for every connected component $\hat{G}$ of $G$, except $G'$, there exists an optimum vertex cover $\hat{X}$ that contains the vertex set $Y \cap V(\hat{G})$. Observe that this holds also for the connected component $G'$ by choosing a vertex $\hat{y} \in Y \setminus Y'$ and considering an optimum vertex cover of $G$ that contains $Y \setminus \{\hat{y}\}$. But, this implies that there exists an optimum vertex cover of $G$ that contains $Y$ which contradicts the assumption that $Y$ is a blocking set of $G$. Hence, the minimal blocking set $Y$ is contained in at most one connected component of $G$.

    Now, we prove item (\ref{enum::bs_basics_sol_size}). Since $Y$ is a blocking set of $G$ it holds by definition that $\OPT(G) < \OPT(G-Y)+|Y|$. For the other direction, let $y \in Y$ be an arbitrary vertex of $Y$. There exists an optimum vertex cover $X$ of $G$ that contains the set $Y \setminus \{y\}$ because $Y$ is a minimal blocking set of $G$. It holds that $X \setminus Y$ is a vertex cover of $G-Y$. Hence, $\OPT(G-Y) \leq |X\setminus Y| = |X|-|Y\setminus \{y\}| = \OPT(G)+1 -|Y|$ which implies that $\OPT(G-Y)+|Y|=\OPT(G)+1$.
\end{proof}

The following lemma shows that deleting a vertex set $Z$ that is part of an optimum solution to a given graph $G$ cannot increase the size of minimal blocking sets. Furthermore, it points out that there is a strong relation between minimal blocking sets in $G$ and $G-Z$.

\begin{lemma} \label{lemma::delete_non_bs}
    Let $G$ be a graph and let $Z \subseteq V(G)$ be a set of vertices such that there exists an optimum vertex cover $X$ of $G$ with $Z \subseteq X$.
    \begin{enumerate}[(i)]
        \item $\OPT(G-Z) +|Z| = \OPT(G)$ \label{enum::delete_non_bs_eq}
        \item If $Y$ is a blocking set of $G$ then $Y\setminus Z$ is a blocking set of $G-Z$. \label{enum::delete_non_bs_dZ}
        \item It holds that $\bsg{G-Z} \leq \bsg{G}$, more precisely, if $Y'$ is a minimal blocking set of $G-Z$ then there exists a (possibly empty) set $Z' \subseteq Z$ such that $Z' \cup Y'$ is a minimal blocking set of $G$. \label{enum::delete_non_bs_aZ}
    \end{enumerate}
\end{lemma}

\begin{proof}
    Obviously, item (\ref{enum::delete_non_bs_eq}) holds, because every optimum vertex cover of $G-Z$ together with $Z$ is a vertex cover of $G$, and because there exists an optimum vertex cover of $G$ that contains $Z$.

    Next, we prove item (\ref{enum::delete_non_bs_dZ}) by contradiction. Let $Y$ be a blocking set of $G$ and assume that $Y\setminus Z$ is not a blocking set of $G-Z$. This implies that there exists an optimum vertex cover $X$ of $G-Z$ that contains the set $Y \setminus Z$. However, $X \cup Z$ is an optimum vertex cover of $G$ that contains $Y$ because $\OPT(G-Z) +|Z| = \OPT(G)$. This contradicts the assumption that $Y$ is a blocking set of $G$ and shows that $Y\setminus Z$ is a blocking set of $G-Z$.

    Finally, we prove item (\ref{enum::delete_non_bs_aZ}). Let $Y'$ be a minimal blocking set of $G-Z$. First, we show that $Y' \cup Z$ is a blocking set of $G$. Assume for contradiction that $Y' \cup Z$ is not a blocking set of $G$. Hence, there exists an optimum vertex cover $X$ of $G$ that contains $Y' \cup Z$. Since $\OPT(G-Z) +|Z| = \OPT(G)$ it holds that $X \setminus Z$ is an optimum vertex cover of $G-Z$ that contains $Y'$. This contradicts the fact that $Y'$ is a blocking set of $G-Z$, and proves that $Y' \cup Z$ is a blocking set of $G$.
    To conclude the proof we will show that every minimal blocking set $Y \subseteq Y' \cup Z$ of $G$ contains the set $Y'$. Let $Y \subseteq Y' \cup Z$ be a minimal blocking set of $G$. It follows from item (\ref{enum::delete_non_bs_dZ}) that $Y \setminus Z \subseteq Y'$ is a blocking set of $G-Z$. Since $Y'$ is a minimal blocking set of $G-Z$ it must hold that $Y \setminus Z = Y'$; thus, it holds that $Y' \subseteq Y$.
\end{proof}

\begin{lemma} \label{lemma::structure_bs}
    Let $G$ be a graph that contains a vertex $r$ that is contained in at least one, but not all, optimum vertex covers of $G$, and is not adjacent to all other vertices of $G$, i.e., $N[r] \subsetneq V(G)$. Let $Y$ be a minimal blocking set of $G$ with $r \notin Y$.\footnote{Such a minimal blocking set exists because $V(G)\setminus \{r\}$ is a blocking set of $G$; otherwise every optimum solution contains all except one vertex of $G$ which implies that $V(G)$ is the only (minimal) blocking set of $G$. But the only graphs that have the entire set of vertices as a minimal blocking set are cliques.}
    \begin{enumerate}[(i)]
        \item The set $Y$ is also a blocking set of $G-r$ and the set $Y\setminus N[r]$ is a blocking set of $G-N[r]$. \label{enum::structure_bs_delete}
        \item Let $Y'$ be a minimal blocking set of $G-r$, and let $\widetilde{Y}$ be a minimal blocking set of $G-N[r]$. The set $Y' \cup \widetilde{Y}$ is a blocking set of $G$. \label{enum::structure_bs_union}
        \item For all vertices $y \in Y$ it holds that $Y \setminus \{y\}$ is not a blocking set of $G-r$ or that $Y \setminus (\{y\} \cup N[r])$ is not a blocking set of $G-N[r]$. \label{enum::structure_bs_necess}
    \end{enumerate}
\end{lemma}

\begin{proof}
    Since $G$ contains an optimum vertex cover that contains $r$ and an optimum vertex cover that contains $N(r)$ it follows from Lemma \ref{lemma::delete_non_bs} item (\ref{enum::delete_non_bs_dZ}) that $Y\setminus \{r\}=Y$ is a blocking set of $G-r$ and that $Y\setminus N(r)=Y\setminus N[r]$ is a blocking set of $G-N(r)$. Note that $Y \setminus N[r]$ is also a blocking set of $G-N[r]$ because vertex $r$ is isolated in $G-N(r)$ and not contained in $Y \setminus N(r)$. This proves item (\ref{enum::structure_bs_delete}).

    Next, we prove item (\ref{enum::structure_bs_union}). Let $Y'$ be a minimal blocking set of $G-r$ and $\widetilde{Y}$ be a minimal blocking set of $G-N[r]$. Assume for contradiction that $\overline{Y}=Y' \cup \widetilde{Y}$ is not a blocking set of $G$. Thus, there exists an optimum vertex cover $X$ of $G$ that contains $\overline{Y}$. If vertex $r$ is contained in $X$, then $X \setminus \{r\}$ is an optimum vertex cover of $G-r$ that contains $\overline{Y}$ and therefore also $Y'$. This contradicts the fact that $Y'$ is a blocking set of $G-r$. If vertex $r$ is not contained in $X$, then $X$ contains the vertex set $N(r)$ and $X \setminus N(r)$ is an optimum vertex cover of $G-N[r]$ that contains $\widetilde{Y} \subseteq \overline{Y}\setminus N(r)$. But this is a contradiction to the assumption that $\widetilde{Y}$ is a blocking set of $G-N[r]$. Overall, this proves that $Y' \cup \widetilde{Y}$ is a blocking set of $G$.

    Finally, we show item (\ref{enum::structure_bs_necess}). Assume for contradiction that there exists a vertex $y \in Y$ such that $Y \setminus \{y\}$ is a blocking set in $G-r$ and such that $Y \setminus (\{y\} \cup N[r])$ is a blocking set in $G-N[r]$. Let $Y' \subseteq Y \setminus \{y\}$ be a minimal blocking set in $G-r$ and let $\widetilde{Y} \subseteq Y \setminus (\{y\} \cup N[r])$ be a minimal blocking set in $G-N[r]$. It follows from item (\ref{enum::structure_bs_union}) that $Y' \cup \widetilde{Y} \subseteq Y \setminus \{y\} \subsetneq Y$ is a blocking set of $G$. This contradicts the assumption that $Y$ is a minimal blocking set of $G$ and concludes the prove of item (\ref{enum::structure_bs_necess}).
\end{proof}

In the remainder of this section, we show upper and lower bounds on the blocking set sizes of graphs with elimination distance $d$ to a graph class \C. For hereditary graph classes, these bounds turn out to be tight. Theorem~\ref{thm:intro:precise-bounds} will follow directly from the lower bound presented in Theorem~\ref{theorem::lb}, combined with the upper bound in Theorem~\ref{theorem::ub_hereditary}.
\subsection{Lower bound}

In this subsection we give a lower bound for $\bsd{\C}{d}$ for all robust graph classes. Our result improves the lower bound for the size of largest minimal blocking sets of graphs with treedepth $d$ given by Bougeret and Sau \cite{BougeretS17}, because a graph that has treedepth $d$ has elimination distance $d-1$ to the graph class where every graph is an independent set. We will show how we can glue two graphs together to obtain larger minimal blocking sets.

\begin{theorem} \label{theorem::lb}
    Let $\C$ be a robust graph class where \bsc{\C} is bounded.
    For every integer $d\geq 1$ it holds:
    \begin{align*}
        \bsd{\C}{d} \geq \begin{cases}
                            2^{d-1}+1 & \text{, if } \bsc{\C} = 1 \\
                            (\bsc{\C}-1) 2^d+1 & \text{, if } \bsc{\C} \geq 2
                        \end{cases}
    \end{align*}
\end{theorem}

To show Theorem \ref{theorem::lb} we construct for each graph class $\C$ where $\bsc{\C}$ is bounded and each integer $d\geq1$ a graph $G$ with $\ed{\C}{G}=d$ that contains a minimal blocking set of size at least $2^{d-1}+1$ when $\bsc{\C} = 1$, and of size at least $(\bsc{\C}-1) 2^d+1$ when $\bsc{\C} \geq 2$. Since $\bsd{\C}{d} = \max\{\bsg{G} \mid \ed{\C}{G}\leq d\}$ this will prove Theorem \ref{theorem::lb}. We use the following two constructions to obtain such a graph $G$.

\begin{lemma} \label{lemma::lb_add_vertex}
    Let $H=(V,E)$ be a graph and let $Y$ be a minimal blocking set of $H$.
    Let $H'=(V \cup \{v\}, E \cup \{ \{v,y\} \mid y \in Y\})$ be the graph that results from $H$ by adding a new vertex $v$ to $H$ and by connecting it to all vertices in $Y$.
    It holds that $\OPT(H') = \OPT(H)+1$, and that $Y'=Y \cup \{v\}$ is a minimal blocking set of $H'$.
\end{lemma}

\begin{proof}
    First, we show that $\OPT(H') = \OPT(H) +1$. It is clear that $\OPT(H') \leq \OPT(H)+1$ because every (optimum) vertex cover of $H$ together with the newly added vertex $v$ is a vertex cover of $H'$. Now, assume for contradiction that $\OPT(H')=\OPT(H)$. This implies that every optimum vertex cover $X$ of $H'$ does not contain the vertex $v$, and therefore the neighborhood $Y$ of $v$ in $H'$. This contradicts the fact that $Y$ is a blocking set in $H$, because $X$ is an optimum vertex cover of $H$ with $Y \subseteq X$. Thus, $\OPT(H') = \OPT(H) +1$.

    Second, we will show that $Y \cup \{v\}$ is a minimal blocking set of $H'$. The set $Y \cup \{v\}$ is a blocking set of $H'$; otherwise there would exist an optimum vertex cover $X'$ of $H'$ with $Y \cup \{v\} \subseteq X'$. But this would imply that $X=X' \setminus \{v\}$ is an optimum vertex cover of $H$ with $Y \subseteq X$ which contradicts the assumption that $Y$ is a (minimal) blocking set of $H$.
    Now, assume that $Y \cup \{v\}$ is not a minimal blocking set of $H'$. Hence, there exists a set $Y' \subsetneq Y \cup \{v\}$ that is a minimal blocking set in $H'$.
    Observe that there is an optimum vertex cover of $H'$ containing $v$ (take a vertex cover in $H$ and add $v$) and an optimum vertex cover containing $N(v)$, by taking a size-$(\OPT(H)+1)$ vertex cover in $H$ containing $Y$. Thereby,  it holds that neither the set $\{v\}$ nor the set $N(v)=Y$ is a blocking set of $H'$; hence, $Y' \setminus \{v\} \neq \emptyset$ and $v \in Y'$.

    The set $Y' \setminus \{v\} \subsetneq Y$ is not a blocking set of $H$ because $Y$ is a minimal blocking set of $H$. Thus, there exists an optimum vertex cover $X$ of $H$ that contains the set $Y'\setminus \{v\}$. But, $X \cup \{v\}$ is an optimum vertex cover of $H'$ that contains $Y'$ which implies that $Y'$ is not a blocking set of $H'$ and contradicts the choice of $Y'$. This proves that $Y \cup \{v\}$ is a minimal blocking set of $H'$ and concludes the proof.
\end{proof}

\begin{observation} \label{observation::bs_monoton}
    It follows from Lemma \ref{lemma::lb_add_vertex} that $\bsd{\C}{d-1} < \bsd{\C}{d}$.
\end{observation}

\begin{lemma} \label{lemma::lb_union}
    Let $H_1=(V_1,E_1)$ and $H_2=(V_2,E_2)$ be two graphs with $\bsg{H_1}\geq 2$ and $\bsg{H_2}\geq 1$.
    Let $Y_1$ be a minimal blocking set of $H_1$ of size at least two, let $Y_2$ be a minimal blocking set of $H_2$, and let $y_1 \in Y_1$ and $y_2 \in Y_2$.
    Let $H=(V_1 \cup V_2, E_1 \cup E_2 \cup \{ \{y_1,y_2\} \})$ be the graph that results from the union of the graphs $H_1$ and $H_2$ by additionally connecting the vertices $y_1$ and $y_2$.
    It holds that $\OPT(H) = \OPT(H_1)+\OPT(H_2)$, and that $Y=(Y_1 \cup Y_2) \setminus \{y_1,y_2\}$ is a minimal blocking set of $H$.
\end{lemma}

\begin{proof}
    Since every (optimum) vertex cover of $H$ contains a vertex cover of $H_1$ and $H_2$ it holds that $\OPT(H) \geq \OPT(H_1)+\OPT(H_2)$.
    The inequality $\OPT(H) \leq \OPT(H_1)+\OPT(H_2)$ follows from the fact that there exists an optimum vertex cover $X_1$ of $H_1$ that contains $y_1$: This holds because $Y_1$ is a minimal blocking set of $H_1$ of size at least two that contains the vertex $y_1$. The optimum vertex cover $X_1$ of $H_1$ that contains $y_1$ together with any optimum vertex cover of $H_2$ is a vertex cover of $H$. Thus, $\OPT(H) \leq \OPT(H_1)+\OPT(H_2)$ which concludes the proof that $\OPT(H) = \OPT(H_1)+\OPT(H_2)$.

    Next, we show that $Y=(Y_1 \cup Y_2) \setminus \{y_1,y_2\}$ is a minimal blocking set of $H$. First, we show that $Y$ is a blocking set of $H$. Assume for contradiction that $Y$ is not a blocking set of $H$. This implies that there exists an optimum vertex cover $X$ of $H$ with $Y \subseteq X$. Since $\{y_1,y_2\}$ is an edge in $H$ it holds that $y_1 \in X$ or $y_2 \in X$; thus, $Y_1 \subseteq X$ or $Y_2 \subseteq X$.
    But, this contradicts the assumption that $Y_1$ is a blocking set of $H_1$ or the assumption that $Y_2$ is a blocking set of $H_2$ respectively because $X \cap V(H_1)$ and $X \cap V(H_2)$ are optimum vertex covers of $H_1$ resp.\ $H_2$, and because at least one of them contains the vertex set $Y_1$ resp.\ $Y_2$. This proves that $Y$ is a blocking set of $H$.

    Second, we show that $Y$ is also a minimal blocking set of $H$.
    Let $Y' \subseteq Y$ be a minimal blocking set of $H$. If $Y' = Y$ we are done, so suppose there exists $y \in Y \setminus Y'$. Suppose $y \in Y_i$ for $i \in \{1,2\}$. Let $X_i$ be an optimal vertex cover in $H_i$ containing $Y_i \setminus \{y\}$. Let $j \in \{1,2\} \setminus \{i\}$ and let $X_{j}$ be an optimal vertex cover in $H_j$ containing $Y_j \setminus \{y_j\}$. Both these vertex covers exist, since $Y_i$ respectively $Y_j$ are minimal blocking sets in $H_i$ respectively $H_j$. Observe that $X_i \cup X_j$ is a vertex cover of $H$, since $X_i$ is a vertex cover in $H_i$, $X_j$ is a vertex cover in $H_j$ and the edge $\{y_1,y_2\}$ is covered by $y_i$. Furthermore, it has minimum size since $\OPT(H) = \OPT(H_1)+\OPT(H_2)$. However, $Y'\subseteq (Y_1 \cup Y_2) \setminus \{y,y_j\} \subseteq X_i \cup X_j$, contradicting that $Y'$ is a blocking set.
\end{proof}

The next lemma shows how we can combine these two constructions to obtain larger minimal blocking sets for a graph of elimination distance $d$ to a graph class \C when given a graph of elimination distance $d-1$ to a graph class \C.

\begin{lemma} \label{lemma::lb_combine}
    Let $H=(V,E)$ be a graph, and let $Y$ be a minimal blocking set of $H$. There exists a graph $H'$ that fulfills the following properties:
    \begin{enumerate}[(i)]
        \item The graph $H'$ has a minimal blocking set of size $2|Y|-1$. \label{enum::lb_bss}
        \item For every graph class $\C$ it holds that $\ed{\C}{H'} \leq \ed{\C}{H}+1$. \label{enum::lb_ed}
    \end{enumerate}
\end{lemma}

To prove Lemma \ref{lemma::lb_combine} we combine the two constructions of Lemma \ref{lemma::lb_add_vertex} and Lemma \ref{lemma::lb_union}.

\begin{proof}
    To construct graph $H'$, we first apply Lemma \ref{lemma::lb_add_vertex} to graph $H$ and the minimal blocking set $Y$. Let $H_1$ be the resulting graph, and let $y_1$ be the vertex that we added to $H$ during the construction of Lemma \ref{lemma::lb_add_vertex} to obtain graph $H_1$. It follows from Lemma \ref{lemma::lb_add_vertex} that $Y_1=Y \cup \{y_1\}$ is a minimal blocking set of $H_1$.

    Second, we apply Lemma \ref{lemma::lb_union} to the two graphs $H_1$ and $H_2=H$, the minimal blocking set $Y_1=Y \cup \{y_1\}$ of $H_1$ of size at least two, the minimal blocking set $Y_2=Y$ of $H_2$, and the vertex $y_1 \in Y_1$ as well as an arbitrary vertex $y_2 \in Y_2$. Let $H'$ be the resulting graph.
    The set $Y' = (Y_1 \cup Y_2)\setminus \{y_1,y_2\}$ is a minimal blocking set of $H'$ (Lemma \ref{lemma::lb_union}). The size of $Y'$ is $|Y_1|+|Y_2|-2=2|Y|-1$. This proves item (\ref{enum::lb_bss}).
    Item (\ref{enum::lb_ed}) holds, because $H'-y_1$ consists of two copies of $H$ that are not connected by any edge.

    \begin{figure}
    \centering
    \includegraphics{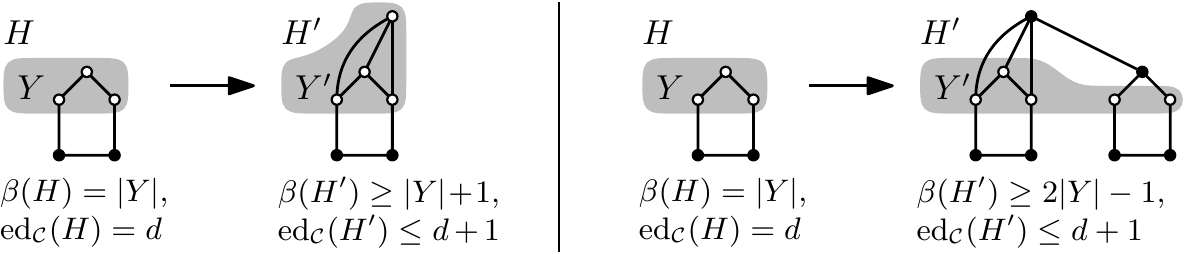}
    \caption{This figure depicts the construction used to prove Lemma~\ref{lemma::lb_add_vertex} (left) and Lemma~\ref{lemma::lb_combine} (right). White vertices form minimal blocking sets.
    }
    \label{fig:blocking-sets-LB}
    \end{figure}

    The combination of these two constructions is depicted in Figure~\ref{fig:blocking-sets-LB}.
\end{proof}

Finally, we are able to prove Theorem \ref{theorem::lb} by induction using the construction of Lemma~\ref{lemma::lb_add_vertex} as well as Lemma~\ref{lemma::lb_combine}. We need the construction of Lemma \ref{lemma::lb_add_vertex} only for the base case of the induction because for a graph with minimal blocking set size one, this construction leads to a graph with minimal blocking set size at least two whereas Lemma \ref{lemma::lb_combine} only leads to a graph of minimal blocking set size one.

\begin{proof}[Proof of Theorem \ref{theorem::lb}]
    We prove Theorem \ref{theorem::lb} by induction over the elimination distance $d$ to graph class \C.
    In the base case of the induction we construct a graph $G$ that has elimination distance at most $d=1$ to graph class \C. Let $H$ be a graph of class \C with $\bsg{H}=\bsc{\C}$, and let $Y$ be a minimal blocking set of $H$ of size $\bsc{\C}$.
    If $\bsc{\C}=1$ then we apply Lemma \ref{lemma::lb_add_vertex} to graph $H$ and the minimal blocking set $Y$ of $H$ of size \bsc{\C} to obtain graph $G$. Let $y$ be the vertex that we add to graph $H$ during the construction of Lemma \ref{lemma::lb_add_vertex} to obtain graph $G$. Obviously, it holds that $\ed{\C}{H_1} \leq 1$. Furthermore, $Y'=Y \cup \{y\}$ is a minimal blocking set of $G$ which implies that $\bsg{G} \geq |Y'| = |Y|+1 = \bsg{H}+1=\bsc{\C}+1 = 2 = 2^{1-1}+1$. Hence, $\bsd{\C}{1} \geq \bsg{G} \geq 2^{1-1}+1$ for $d=1$ and $\bsc{\C}=1$ which is the desired bound for Theorem \ref{theorem::lb}.

    In the case that $\bsc{\C}\geq 2$ we apply Lemma \ref{lemma::lb_combine} to the graph $H$ and the minimal blocking set $Y$ of $H$ of size \bsc{\C} to obtain graph $G$. It follows from Lemma \ref{lemma::lb_combine} item (\ref{enum::lb_ed}) that $\ed{\C}{G} \leq \ed{\C}{H}+1 = 1$ and from item (\ref{enum::lb_bss}) that $\bsd{\C}{d} \geq \bsg{G} \geq 2 |Y| - 1 = 2 \bsc{\C} - 1 = (\bsc{\C}-1) \cdot 2^1 + 1$. This concludes the proof of Theorem \ref{theorem::lb} for $d=1$.

    For the induction step, we assume that the statement is true for all integers less than $d$, and all graph classes \C where $\bsc{\C}$ is bounded.

    By induction there exists a graph $H$ with $\ed{\C}{H}\leq d-1$ such that
    \begin{align} \label{align::lb}
        \bsg{H} \geq \begin{cases}
                        2^{(d-1)-1}+1 & \text{, if } \bsc{\C}=1 \\
                        (\bsc{\C}-1)2^{(d-1)}+1 & \text{, if } \bsc{\C}\geq 2
                    \end{cases}
    \end{align}
    because $\bsd{\C}{d-1}$ has at least this size.
    Let $Y$ be a minimal blocking set of $H$ of size \bsg{H}. Again, we apply Lemma \ref{lemma::lb_combine} to the graph $H$ and the minimal blocking set $Y$. It follows from Lemma \ref{lemma::lb_combine} item (\ref{enum::lb_ed}) that $\ed{\C}{G} \leq \ed{\C}{H}+1 = d$ and from item (\ref{enum::lb_bss}) that $\bsg{G} \geq 2 |Y| - 1 = 2 \bsg{H} - 1$. Using the induction hypothesis (inequality \ref{align::lb}) we obtain:
    \begin{align*}
        \bsg{G} \geq 2 \bsg{H} - 1 \overset{(\ref{align::lb})}{\geq}
            \begin{cases}
                2 \cdot \left(2^{(d-1)-1}+1 \right) -1 = 2^{d-1} + 1 & \text{, if } \bsc{\C}=1 \\
                2 \cdot \left( (\bsc{\C}-1)2^{(d-1)}+1 \right) -1 = (\bsc{\C}-1)2^d+1  & \text{, if } \bsc{\C}\geq 2. \\
            \end{cases}
    \end{align*}
    This concludes the proof of Theorem \ref{theorem::lb} because $\bsd{\C}{d} = \max\{\bsg{G} \mid \ed{\C}{G} \leq d\}$.
\end{proof}

\subsection{Upper bound}

We will show that our lower bound for $\bsd{\C}{d}$, which we proved in Theorem \ref{theorem::lb}, is tight when the graph class \C is hereditary. As mentioned above, this gives us also a tight bound for the size of a largest minimal blocking sets for graphs with treedepth $d$. Thus, our result improves the upper and lower bound for the largest minimal blocking set size of graphs of treedepth (at most) $d$ \cite{BougeretS17}.
Afterwards, we will show an upper bound for $\bsd{\C}{d}$ for some non-hereditary graph classes, including the class of graphs where the optimum solution is equal to the optimum \LP solution. This bound is weaker than the bound for $\bsd{\C}{d}$ when $\C$ is hereditary.

\begin{theorem} \label{theorem::ub_hereditary}
    Let $\C$ be a robust hereditary graph class where \bsc{\C} is bounded, and let $d\geq 0$.
    \begin{align*}
        \bsd{\C}{d} \leq \begin{cases}
                            1 & \text{, if } \bsc{\C} = 1 \text{ and } d=0 \\
                            2^{d-1}+1 & \text{, if } \bsc{\C} = 1 \text{ and } d\geq1 \\
                            (\bsc{\C}-1) 2^d+1 & \text{, if } \bsc{\C} \geq 2
                        \end{cases}
    \end{align*}
\end{theorem}

We consider different cases in our proof. But, only one case needs the requirement that the class \C is hereditary. We will see later that we can bound the blocking set size also for some non-hereditary class \C for the remaining case.

\begin{lemma} \label{lemma::ub_r_in_bs}
    Let $G$ be a graph, let $Y$ be a minimal blocking set in $G$, and let $r \in Y$ be an arbitrary vertex that is contained in the blocking set $Y$.
    \begin{enumerate}[(i)]
        \item If $\OPT(G)=\OPT(G-r)$ then $Y=\{r\}$. \label{enum::ub_r_in_bs_same}
        \item If $\OPT(G)=\OPT(G-r)+1$ then $Y\setminus \{r\}$ is a minimal blocking set of $G-r$. \label{enum::ub_r_in_bs_plus}
    \end{enumerate}
\end{lemma}

\begin{proof}
    If $\OPT(G)=\OPT(G-r)$ then no optimum vertex cover of $G$ contains $r$. Thus, $\{r\}$ is a (minimal) blocking set of $G$ which implies that $Y=\{r\}$ and proves item (i).

    Now, assume that $\OPT(G)=\OPT(G-r)+1$. Consequently, there exists an optimum vertex cover of $G$ that contains vertex $r$ which implies that $Y \neq \{r\}$. Let $Y'=Y \setminus \{r\} \neq \emptyset$.
    It follows from Lemma \ref{lemma::delete_non_bs} item (\ref{enum::delete_non_bs_dZ}) that $Y'$ is a blocking set of $G-r$. Furthermore, it follows from Lemma \ref{lemma::delete_non_bs} item (\ref{enum::delete_non_bs_aZ}) that $Y'$ is a minimal blocking set of $G-r$, because every minimal blocking set $\hat{Y} \subseteq Y'$ of $G-r$ can be extended to the blocking set $\hat{Y} \cup \{r\} \subseteq Y$ of $G$. Since $Y$ is a minimal blocking set of $G$, the set $Y'$ must be a minimal blocking set of $G-r$.
\end{proof}

\begin{lemma} \label{lemma::ub_r_not_bs_OPT}
    Let $G$ be a graph that contains a vertex $r$ that is not contained in any optimum vertex cover of $G$; hence $\OPT(G-r)=\OPT(G)$. Let $Y$ be a minimal blocking set of $G$ with $r \notin Y$.\footnote{Such a blocking set exists because $V(G) \setminus \{r\}$ is a blocking set of $G$; otherwise there would exists an optimum vertex cover of $G$ that contains $r$ because every optimum vertex cover of $G$ would contain all except one vertex.}
    There exists a (possibly empty) set $Z \subseteq N(r)$ such that $Y \cup Z$ is a minimal blocking set of $G-r$.
\end{lemma}

\begin{proof}
    Observe that $Y$ may not be a blocking set of $G-r$ because $r$ forces the vertices in $N(r)$ to be in every vertex cover. However, the set $Y \cup N(r)$ is a blocking set of $G-r$; otherwise there exists an optimum vertex cover $X$ of $G-r$ that contains $Y \cup N(r)$. But, the set $X$ is also an optimum vertex cover of $G$ that contains $Y$ which contradicts the assumption that $Y$ is a blocking set of $G$.

    Let $Y' \subseteq Y \cup N(r)$ be a minimal blocking set of $G-r$. Note that $Y' \setminus N(r) \neq \emptyset$ because every optimum vertex cover of $G$ contains $N(r)$ which implies that there exists an optimum vertex cover of $G-r$ that contains $N(r)$. We will show that the minimal blocking set $Y'$ must contain the set $Y$, by proving that the non empty set $Y'\setminus N(r)$ is a blocking set of $G$. If the set $Y' \setminus N(r)$ is not a blocking set of $G$ then there exists an optimum vertex cover $X$ of $G$ that contains the set $Y' \setminus N(r)$. However, the set $X$ contains also the vertex set $N(r)$ because every optimum vertex cover of $G$ does not contain $r$. Thus, $Y' \subseteq X$ and $X$ is an optimum vertex cover of $G-r$ which contradicts the fact that $Y'$ is a blocking set of $G-r$. Now, the set $Z= Y' \cap N(r)$ fulfills the desired properties.
\end{proof}

\begin{lemma} \label{lemma::ub_r_not_bs_OPT+}
    Let $G$ be a graph that contains a vertex $r$ that is contained in every optimum vertex cover of $G$; hence $\OPT(G-r)+1=\OPT(G)$. Let $Y$ be a minimal blocking set of $G$.
    The set $Y$ is also a minimal blocking set of $G-r$.
\end{lemma}

\begin{proof}
    Since $r$ is contained in every optimum vertex cover of $G$ it follows that $r$ is not contained in any minimal blocking set of $G$, and therefore $r \notin Y$ (Proposition \ref{proposition::bs_basics} item (\ref{enum::bs_basics_all})). Furthermore, it follows from Lemma \ref{lemma::delete_non_bs} item (\ref{enum::delete_non_bs_dZ}) that $Y$ is also a blocking set of $G-r$. Item (\ref{enum::delete_non_bs_aZ}) of Lemma \ref{lemma::delete_non_bs} implies that $Y$ is also a minimal blocking set of $G-r$ because no minimal blocking set of $G$ contains vertex $r$.
\end{proof}

We can apply the following Lemma in our proof of Theorem \ref{theorem::ub_hereditary} only when our class \C has some extra properties, because we need to bound the size of minimal blocking sets of graphs where we delete more than the root of the elimination tree; these vertices could be part of the base components of the elimination tree.

\begin{lemma} \label{lemma::ub_split_bs}
    Let $G$ be a graph that contains a vertex $r$ that is contained in at least one, but not all, optimum vertex covers of $G$, and that is not adjacent to all other vertices of $G$; hence $\OPT(G-r)+1=\OPT(G)$ and $N[r] \subsetneq V(G)$. Let $Y$ be a minimal blocking set of $G$ with $r \notin Y$.\footnote{The existence follows from Lemma \ref{lemma::structure_bs}.}
    At least one of the following three cases holds:
    \begin{enumerate}
        \item The set $Y$ is also a minimal blocking set of $G-r$. \label{enum::ub_split_bs_r}
        \item The set $Y=Y\setminus N(r)$ is also a minimal blocking set of $G-N[r]$. \label{enum::ub_split_bs_Nr}
        \item There exists a minimal blocking set $Y' \subsetneq Y$ of $G-r$, and the set $\hat{Y} = Y \setminus Y'$ is a minimal blocking set of $G-Y'$.
        Furthermore, $\hat{Y}$ is not a minimal blocking set of $G-Y'-r$, however, there exists a nonempty set $Z \subseteq N(r)\setminus Y'$ such that $\hat{Y} \cup Z$ is a minimal blocking set of $G-Y'-r$. \label{enum::ub_split_bs_two}
    \end{enumerate}
\end{lemma}

\begin{proof}
    It follows from Lemma \ref{lemma::structure_bs} item (\ref{enum::structure_bs_delete}) that $Y$ is a blocking set of $G-r$ and that $Y \setminus N[r]$ is a blocking set of $G-N[r]$. If $Y$ is also a minimal blocking set of $G-r$ then Case \ref{enum::ub_split_bs_r} holds, and if $Y \setminus N[r] =Y$ is a minimal block set of $G-N[r]$ then Case \ref{enum::ub_split_bs_Nr} holds. Thus, in the following we assume that neither Case \ref{enum::ub_split_bs_r} nor Case \ref{enum::ub_split_bs_Nr} holds.

    Let $Y' \subsetneq Y$ be a minimal blocking set of $G-r$ and let $\widetilde{Y} \subsetneq Y \setminus N(r)$ be a minimal blocking set of $G-N[r]$. Consider the graph $G-Y'$ and the set $\hat{Y}=Y \setminus Y'$. Since $Y' \subsetneq Y$ and since $Y$ is a minimal blocking set of $G$, there exists an optimum vertex cover $X$ of $G$ that contains $Y'$. Therefore, it follows from Lemma \ref{lemma::delete_non_bs} that $\OPT(G-Y')+|Y'| = \OPT(G)$ (item~(\ref{enum::delete_non_bs_eq})) and that $\hat{Y}$ is a minimal blocking set of $G-Y'$ (item~(\ref{enum::delete_non_bs_dZ}) and item~(\ref{enum::delete_non_bs_aZ})).

    It remains to show the second part, namely, that $\hat{Y}$ is not a minimal blocking set of $G-r-Y'$. First, we show that $\OPT(G-Y')= \OPT(G-r-Y')$. Let $\hat{X}$ be an optimum vertex cover of $G-r-Y'$. It holds that $\hat{X} \cup Y'$ is a vertex cover of $G-r$. Since $\hat{X} \cup Y'$ is a vertex cover of $G-r$ that contains $Y'$, and $Y'$ is a blocking set of $G-r$ it holds that $|\hat{X}|+|Y'| \geq \OPT(G-r)+1$. Thus, the optimum vertex $\hat{X}$ of $G-r-Y'$ has size at least $\OPT(G-r)+1-|Y'| = \OPT(G-Y')$. This concludes the proof that $\OPT(G-Y')= \OPT(G-r-Y')$ because it always holds that $\OPT(G-Y') \geq \OPT(G-r-Y')$.

    Now, we are able to prove that $\hat{Y}$ is not a blocking set of $G-r-Y'$. Let $y \in Y' \setminus \widetilde{Y}$. Note that $Y' \setminus \widetilde{Y}$ is not empty because $Y' \cup \widetilde{Y}=Y$ (Lemma \ref{lemma::structure_bs} item (\ref{enum::structure_bs_union})) and because $\widetilde{Y} \subsetneq Y$. Observe that the set $Y \setminus\{y\}$ is not a blocking set of $G-r$ since $Y \setminus\{y\} \supseteq \widetilde{Y}$ is a blocking set of $G-N[r]$ and since Lemma \ref{lemma::structure_bs} item (\ref{enum::structure_bs_necess}) states that $Y \setminus\{y\}$ cannot be a blocking set in $G-r$ and $G-N[r]$.
    Thus, there exists an optimum vertex cover $X'$ of $G-r$ with $\hat{Y} \subseteq Y \setminus \{y\}\subseteq X'$. The set $X=X' \cup \{r\}$ is an optimum vertex cover of $G$ because $\OPT(G)=\OPT(G-r)+1$. Furthermore, the set $X$ contains the vertex set $\hat{Y}$ as well as the vertex $r$, but not the vertex $y$ (because $Y$ is a blocking set of $G$ and $Y \setminus \{y\} \subseteq X$).
    Let $\hat{X} = X \setminus Y'$. It holds that $\hat{X}$ is a vertex cover of $G-Y'$ that contains set $\hat{Y}$ and vertex $r$. Since $Y' \setminus \{y\} \subseteq X$, the set $\hat{X}$ has size $\OPT(G)-|Y'|+1=\OPT(G-Y')+1$. Thus, the set $\hat{X} \setminus \{r\}$ is an optimum vertex cover of $G-r-Y'$ because $\OPT(G-Y')=\OPT(G-r-Y')$, and the set $\hat{X} \setminus \{r\}$ contains the set $\hat{Y}$. Hence, $\hat{Y}$ is not a blocking set of $G-r-Y'$.

    Observe that the graph $G-Y'$ together with the set $\hat{Y}$ fulfills the properties of Lemma \ref{lemma::ub_r_not_bs_OPT}. Hence, there exists a set $Z \subseteq N_{G-Y'}(r)=N(r) \setminus Y'$ such that $\hat{Y} \cup Z$ is a minimal blocking set of $G-Y'-r$. Since, $\hat{Y}$ is not a blocking set in $G-Y'-r$ it holds that $Z \neq \emptyset$. This concludes the proof.
\end{proof}

Combining Lemma \ref{lemma::ub_r_in_bs} to Lemma \ref{lemma::ub_split_bs} we can now prove Theorem \ref{theorem::ub_hereditary}.

\begin{proof}[Proof of Theorem \ref{theorem::ub_hereditary}]
    We prove Theorem \ref{theorem::ub_hereditary} by induction over the integer $d$.
    In the base case assume that $d=0$. Since every graph $G$ with $\ed{\C}{G}=0$ is a union of graphs of the class \C, and every minimal blocking set of $G$ is contained in at most one connected component of $G$ (Proposition \ref{proposition::bs_basics} item (\ref{enum::bs_basics_one_cc})) it follows that $\bsd{\C}{0} = \max\{\bsg{G} \mid \ed{\C}{G}=0\} \leq \bsc{\C}$. Hence, the base case holds for $\bsc{\C}=1$. For $\bsc{\C} \geq 2$ it holds that $(\bsc{\C}-1) \cdot 2^d +1 = \bsc{\C} \geq \bsg{G}$; hence the base case holds for all graph classes where \bsc{\C} is bounded.

    For the induction step, we assume that the statement is true for all integers less than $d$, and all graph classes \C where $\bsc{\C}$ is bounded.

    Let $\C$ be a graph class where $\bsc{\C}$ is bounded, let $G$ be any graph with $\ed{\C}{G} = d$, and let $Y$ be a minimal blocking set of $G$. We can assume that graph $G$ is connected because $Y$ is contained in at most one connected component of $G$ (Proposition \ref{proposition::bs_basics} item (\ref{enum::bs_basics_one_cc})) which implies that $\bsg{G}=\max\{\bsg{G'} \mid G' \text{ connected component of } G\}$. Thus, to bound $\bsd{\C}{d}$ it is enough to bound $\bsg{G}$ for every connected graph $G$ with $\ed{\C}{G}\leq d$.
    Let $r$ be the root of the elimination tree of $G$ to graph class \C. (There exists exactly one root because $G$ is connected.) We will show that $Y$ has the requested size by distinguishing between five cases:

    \begin{description}
        \item[Case 1:] Assume that $\OPT(G)=\OPT(G-r)$ and that $r \in Y$.

        It follows from Lemma \ref{lemma::ub_r_in_bs} item (\ref{enum::ub_r_in_bs_same}) that $Y=\{r\}$. Since $2^{d-1}+1 \geq 1$ and $(\bsc{\C}-1) \cdot 2^d +1 \geq 1$ for $\bsc{\C}\geq 2$, the set $Y$ has the required size.

        \item[Case 2:] Assume that $\OPT(G)=\OPT(G-r)+1$ and that $r \in Y$.

        Observe that the graph $G$ and the set $Y$ fulfill the requirements of Lemma \ref{lemma::ub_r_in_bs} item (\ref{enum::ub_r_in_bs_plus}). Thus, the set $Y\setminus \{r\}$ is a minimal blocking set of $G-r$. Since $r$ was the root of the elimination tree of $G$ it holds that $\ed{\C}{G-r} = d-1$. Hence, we can bound the size of $Y$ by $\bsg{G-r} +1 \leq \bsd{\C}{d-1} + 1$.

        \item[Case 3:] Assume that $\OPT(G)=\OPT(G-r)$ and that $r \notin Y$.

        The graph $G$ together with the set $Y$ and the vertex $r$ fulfills the conditions of Lemma \ref{lemma::ub_r_not_bs_OPT}. This implies that there exists a (possibly empty) set $Z \subseteq N(r)$ such that $Y \cup Z$ is a minimal blocking set of $G-r$. Since $G-r$ has elimination distance $d-1$ to class \C we can bound the size of $Y$ by $\bsg{G-r} \leq \bsd{\C}{d-1}$.

        \item[Case 4:] Assume that $\OPT(G)=\OPT(G-r)+1$, that vertex $r$ is contained in every optimum vertex cover of $G$, and that $r \notin Y$.

        Note that graph $G$, vertex $r$ and the minimal blocking set $Y$ fulfill the requirements of Lemma \ref{lemma::ub_r_not_bs_OPT+}. Therefore, the set $Y$ is also a minimal blocking set of $G-r$. This implies that $|Y|\leq \bsd{\C}{d-1}$ because $G-r$ has elimination distance $d-1$ to class \C.

        \item[Case 5:] Assume that $\OPT(G)=\OPT(G-r)+1$, that vertex $r$ is contained in at least one, but not all, optimum vertex covers of $G$, and that $r \notin Y$.

        Observe that $N[r] \subsetneq V(G)$ because there exists a solution that does not contain $r$ which implies that this solution contains $V(G) \setminus \{r\}$. This in turn would imply that $Y \subseteq V(G) \setminus \{r\}$ is not a blocking set of $G$ which contradicts the assumption.

        This is the only case where we need that the class \C is hereditary, because we have to bound the size of a minimal blocking set in a subgraph of $G-r$.
        The graph $G$, the vertex $r$ and the set $Y$ fulfill the conditions of Lemma \ref{lemma::ub_split_bs}. If Case \ref{enum::ub_split_bs_r} of Lemma \ref{lemma::ub_split_bs} holds then $|Y| \leq \bsg{G-r} \leq \bsd{\C}{d-1}$.

        Next, assume that Case \ref{enum::ub_split_bs_Nr} of Lemma \ref{lemma::ub_split_bs} holds. Since graph class \C is hereditary and deleting vertices of the elimination tree of $G$ can only decrease the elimination distance, it follows that $\ed{\C}{G-N[r]} \leq d-1$. Thus, we can bound the size of the set $Y$, which is a minimal blocking set of $G-N[r]$ by $\bsg{G-N[r]} \leq \bsd{\C}{d-1}$ (Observation \ref{observation::bs_monoton}).

        Now, assume that Case \ref{enum::ub_split_bs_two} of Lemma \ref{lemma::ub_split_bs} holds. Recall that $Y=Y' \dot\cup \hat{Y}$ where $Y'$ is a minimal blocking set in $G-r$ and $\hat{Y}=Y \setminus Y'$ is a minimal blocking set in $G-Y'$. We can bound the size of $Y'$ by $\bsg{G-r} \leq \bsd{\C}{d-1}$. The size of $\hat{Y}$ is at most $\bsg{G-Y'-r}-1$ because there exists a nonempty set $Z \subseteq N(r) \setminus Y'$ such that $\hat{Y} \cup Z$ is a minimal blocking set of $G-Y'-r$. Since $G-Y'-r$ is a subgraph of $G-r$ and the graph class \C is hereditary it holds that $\ed{\C}{G'-Y'-r} \leq d-1$ which implies that $|\hat{Y}|\leq \bsg{G-Y'-r}-1 \leq \bsd{\C}{d-1}-1$ (Observation \ref{observation::bs_monoton}). Combining the bound for $Y'$ and $\hat{Y}$ we obtain that $|Y| \leq 2 \bsd{\C}{d-1} -1$
    \end{description}

    Overall, we showed that $|Y| \leq \max\{1,\bsd{\C}{d-1}+1,2\bsd{\C}{d-1}-1\}$.
    For $d=1$ and $\bsc{\C}=1$ this implies that $|Y|\leq 2 = 2^{1-1}+1$ because $\bsd{\C}{0}=\bsc{\C}=1$. Thus, $\bsd{\C}{1} \leq 2^{1-1}+1$ because we picked an arbitrary graph $G$ with $\ed{\C}{G}=1$, and an arbitrary minimal blocking set of $G$.

    If we assume that $d\geq 2$ or that $\bsc{\C}\geq 2$, then $\max\{1,\bsd{\C}{d-1}+1,2\bsd{\C}{d-1}-1\} = 2\bsd{\C}{d-1}-1$ because $\bsd{\C}{d-1} \geq 2$.
    Hence, we can bound $Y$ by $2\bsd{\C}{d-1}-1$ when $d\geq2$ or $\bsc{\C}\geq 2$. Furthermore, we can bound $\bsd{\C}{d}$ by $2\bsd{\C}{d-1}-1$ when $d\geq2$ or $\bsc{\C}\geq 2$ because we picked an arbitrary graph $G$ with $\ed{\C}{G}=d$ as well as an arbitrary minimal blocking set of $G$.

    For $\bsc{\C}=1$ and $d \geq 2$ this leads to the following upper bound for \bsd{\C}{d}:
    \[
        \bsd{\C}{d} \leq 2 \cdot \bsd{\C}{d-1} \leq 2 \cdot \left(2^{(d-1)-1}+1\right) - 1 = 2^{d-1}+1
    \]
    The second inequality follows from the induction hypothesis.
    Finally, we show the upper bound for \bsd{\C}{d} for the case that $\bsc{\C}\geq 2$ (and $d \geq 1$):
    \[
        \bsd{\C}{d} \leq 2 \cdot \bsd{\C}{d-1} \leq 2 \cdot \left( (\bsc{\C}-1) 2^{(d-1)}+1\right) - 1 = (\bsc{\C}-1) 2^d+1
    \]
    Again, the second inequality follows from the induction hypothesis.
\end{proof}

It follows from Theorem \ref{theorem::lb} and Theorem \ref{theorem::ub_hereditary} that the bound for $\bsd{\C}{d}$ is tight for all hereditary graph classes $\C$, proving Theorem~\ref{thm:intro:precise-bounds}.

\begin{thm:theorem4}
\theoremfour
\end{thm:theorem4}

In the remaining part of this section, we will show that we can obtain an upper bound for \bsd{\C}{d} even if \C is not hereditary, but fulfills some other additional properties. Nevertheless, this leads to a weaker upper bound.

\begin{definition}
    We say that a graph class \C is \emph{$f$-robust}, if $\bsc{\C+c} \leq f(c) = f(\bsc{\C},c)$ for a computable function~$f$.
\end{definition}

\begin{observation} \label{observation::bound_bs_bc}
    It holds that $\bsc{\C} < \bsc{\C+1}$ for all classes $\C$ (Observation \ref{observation::bs_monoton}). This implies that $\bsc{\C+c} \geq 2$ for all $c \geq 2$.
\end{observation}

\begin{theorem} \label{theorem::ub_nonhereditary}
    Let $\C$ be an $f$-robust and robust graph class where \bsc{\C} is bounded, and let $d \geq 0$. It holds that
    \[
        \bsd{\C+c}{d} \leq \bigg(\sum_{i=0}^d \binom{d}{i} f(c+i)\bigg) - 2^d+1.
    \]
\end{theorem}

To prove Theorem \ref{theorem::ub_nonhereditary} we need an additional lemma that allows us to bound the largest minimal blocking set size of subgraphs of $G-r$, where $G$ is a connected graph with $\ed{\C}{G}=d$ and where $r$ is the root of the elimination tree of $G$ to graph class $\C$.

\begin{lemma} \label{lemma::nonhereditary_delete}
    Let $G$ be a connected graph with $\ed{\C+c}{G}=d$, let $r$ be the root of the elimination tree of $G$ to graph class $\C+c$ with $c \geq 0$, and let $Z \subseteq V(G)$ such that $\OPT(G-Z)+|Z|=\OPT(G)$.
    It holds that $\bsg{G-Z-r} \leq \bsd{\C+c+1}{d-1}$.
\end{lemma}

\begin{proof}
    We construct a graph $\hat{G}$ with $\ed{\C+c+1}{\hat{G}} = d-1$ that has the property that for every minimal blocking set $Y$ of $G-Z-r$ there exists a (possibly empty) set $\hat{Z} \subseteq Z$ such that $Y \cup \hat{Z}$ is a minimal blocking set of $\hat{G}$.

    To construct graph $\hat{G}$ from graph $G$ we first delete all vertices from $Z$ that are part of the elimination tree of $G$ to graph class $\C+c$. Let $Z' \subseteq Z$ be the vertices of $Z$ that are contained in the base components of the elimination tree of $G$ to graph class $\C+c$. Observe that $\ed{\C+c}{G-Z\setminus Z'}\leq d$ because we only delete vertices of the elimination tree. Second, let $\mathcal{H}$ be the set of base components of the elimination tree of $G$ to graph class $\C+c$.
    For every graph $H \in \mathcal{H}$ with $V(H) \cap Z' \neq \emptyset$ we add a vertex $v_H$ to graph $H$ and connect it to all vertices in $Z' \cap V(H)$. The resulting graph is $\widetilde{G}$.
    Note that we add at most one vertex to each graph in $\mathcal{H}$. Thus, every graph in $\mathcal{H}$ belongs now to the class $\C+c+1$ which implies that $\ed{\C+c+1}{\widetilde{G}} = d$ (because we can use the same elimination tree as for $G-Z\setminus Z'$). Finally, we delete vertex $r$ form $\widetilde{G}$ to obtain $\hat{G}$. It holds that $\ed{\C+c+1}{\hat{G}} = d-1$, because $r$ is the root of the elimination tree of $\widetilde{G}$ to graph class $\C+c+1$.

    First, we show that $\OPT(\hat{G})=\OPT(G-Z-r)+|Z'|$. Let $X$ be an optimum vertex cover of $G-Z-r$. The set $X \cup Z'$ is a vertex cover of $\hat{G}$ because $G-(Z\setminus Z')-r$ is a subgraph of $\hat{G}$ and because every newly added vertex is only adjacent to vertices in $Z'$. Hence, $\OPT(\hat{G}) \leq \OPT(G-Z-r)+|Z'|$.

    For the other direction, let $\hat{X}$ be an optimum vertex cover of $\hat{G}$.
    If $Z' \subseteq \hat{X}$ then $\hat{X} \setminus Z'$ is a vertex cover of $G-Z-r$ of size $\OPT(\hat{G})-|Z'|$. Note that $\hat{X}$ does not contain any newly added vertex in this case because these vertices are only adjacent to vertices in $Z'$.
    If $\hat{X}$ contains any newly added vertex then we add $r$ as well as $Z \setminus Z'$ to $\hat{X}$ and delete all newly added vertices from $\hat{X}$. The resulting set is a vertex cover of $G$ of size at most $\OPT(\hat{G})+|Z\setminus Z'|$ because $G-(Z\setminus Z')-r$ is a subgraph of $\hat{G}$ and because $\hat{X}$ contains at least one newly added vertex. Together with the assumption that $\OPT(G)=\OPT(G-Z)+|Z| \geq \OPT(G-Z-r) +|Z|$ it follows that $\OPT(G-Z-r) \leq \OPT(G) - |Z| \leq \OPT(\hat{G}) + |Z \setminus Z'| - |Z| = \OPT(\hat{G}) - |Z'|$. Overall, we showed that $\OPT(\hat{G})=\OPT(G-Z-r)+|Z'|$.

    Let $Y$ be a minimal blocking set of $G-Z-r$. We will show that $Y \cup Z'$ is a blocking set of $\hat{G}$. Assume for contradiction that $Y \cup Z'$ is not a blocking set of $\hat{G}$, and let $\hat{X}$ be an optimum vertex cover of $\hat{G}$ that contains $Y \cup Z'$. This implies that $\hat{X} \setminus Z'$ is an optimum vertex cover of $G-Z-r$ that contains $Y$ because $G-Z-r$ is a subgraph of $\hat{G}$, because $Z' \subseteq \hat{X}$, and because $\OPT(\hat{G})=\OPT(G-Z-r)+|Z'|$. This contradicts the assumption that $Y$ is a (minimal) blocking set of $G-Z-r$ and proves that $Y \cup Z'$ is a blocking set of $\hat{G}$.

    Now, let $\hat{Y} \subseteq Y \cup Z'$ be a minimal blocking set of $\hat{G}$. We will show that $Y \subseteq \hat{Y}$. This directly implies that $\bsg{G-Z-r}\leq \bsg{\hat{G}} \leq \bsd{\C+c+1}{d-1}$. Observe that $Z'$ is not a blocking set of $\hat{G}$ because there exists an optimum vertex cover of $\hat{G}$ that contains $Z'$, namely every optimum vertex cover of $G-Z-r$ together with the set $Z'$; hence $\hat{Y}\setminus Z' \neq \emptyset$.
    If $\hat{Y} \setminus Z' \subsetneq Y$ then $\hat{Y} \setminus Z'$ is not a blocking set of $G-Z-r$ because $Y$ is a minimal blocking set of $G-Z-r$. Thus, there exists an optimum vertex cover $X$ of $G-Z-r$ which contains the set $\hat{Y} \setminus Z' \subsetneq Y$. But, $\hat{X} = X \cup Z'$ is a vertex cover of $\hat{G}$, because $G-(Z\setminus Z')-r$ is a subgraph of $\hat{G}$ and every newly added vertex is only adjacent to vertices in $Z'$. Furthermore, $\hat{X}$ is an optimum vertex cover of $\hat{G}$ because $|\hat{X}|=|X| + |Z'| = \OPT(G-Z-r)+|Z'| = \OPT(\hat{G})$, and it holds that $\hat{Y} \subseteq \hat{X}$. This contradicts the assumption that $\hat{Y}$ is a (minimal) blocking set of $\hat{G}$, and concludes the proof.
    \end{proof}

\begin{proof}[Proof of Theorem \ref{theorem::ub_nonhereditary}]
    The proof of Theorem \ref{theorem::ub_nonhereditary} is very similar to the proof of Theorem \ref{theorem::ub_hereditary}. We prove Theorem \ref{theorem::ub_nonhereditary} also by induction over the integer $d$, and for every integer $d$ we show that for all graphs $G$ with $\ed{\C+c}{G}=d$, and every minimal blocking set $Y$ of $G$, we can bound the size of $Y$ by the claimed upper bound. Furthermore, we use the same definition by cases as in the proof of Theorem \ref{theorem::ub_nonhereditary}.

    In the base case assume that $d=0$. Every graph $G$ with $\ed{\C+c}{G} = 0$ is a graph of the class $\C+c$. Therefore, it holds that $\bsg{G}\leq \bsc{\C+c} = f(c) = \sum_{i=0}^0 \binom{0}{i} f(c+i) - 2^0+1$.

    For the inductive step we assume that the statement holds for all integers less than $d$, and for all $f$-robust graph classes $\C+c$ with $c \geq 0$.

    Let \C be any $f$-robust graph class where $\bsc{\C}$ is bounded, let $G$ be any graph with $\ed{\C+c}{G}=d$, and let $Y$ be a minimal blocking set of $G$. As in the proof of Theorem \ref{theorem::ub_hereditary} we can assume that graph $G$ is connected. Let $r$ be the root of the elimination tree of $G$ to graph class $\C+c$. We show that $Y$ has the desired size by distinguishing between the same five cases as in Theorem \ref{theorem::ub_hereditary}:

    Observe, Case 1 to 4 of the Proof of Theorem \ref{theorem::ub_hereditary} never used that the graph class is hereditary. We only used that we can bound $\bsg{G-r}$ using our induction hypothesis. Thus, even when graph class $\C+c$ is not hereditary, we obtain the same bound for the size of a minimal blocking set $Y$ in these cases, because $\ed{\C+c}{G-r}=d-1$. Recall that Case 2 leads to the worst upper bound for the size of $Y$, namely $|Y|\leq \bsg{G-r}+1 \leq \bsd{\C+c}{d-1}+1$.

    Now, we consider the remaining case where we assumed that the class \C is hereditary.
    \begin{description}
        \item[Case 5:] Assume that $\OPT(G)=\OPT(G-r)+1$, that vertex $r$ is contained in at least one, but not all, optimum vertex covers of $G$, and that $r \notin Y$.

        The graph $G$, the vertex $r$, and the set $Y$ fulfill the requirements of Lemma \ref{lemma::ub_split_bs}.\footnote{It holds that $N[r] \subsetneq V(G)$; otherwise $Y$ would not be a blocking set of $G$ (see proof of Theorem \ref{theorem::ub_hereditary} Case~5.} As in the proof of Theorem \ref{theorem::ub_hereditary} we distinguish between the three Cases of Lemma \ref{lemma::ub_split_bs}.
        If Case \ref{enum::ub_split_bs_r} of Lemma \ref{lemma::ub_split_bs} holds then $Y$ is also a minimal blocking set of $G-r$ which implies that $|Y|\leq \bsg{G-r} \leq \bsd{\C+c}{d-1}$.
        Next, assume that Case \ref{enum::ub_split_bs_Nr} of Lemma \ref{lemma::ub_split_bs} holds. Thus, the set $Y=Y \setminus N[r]$ is also a minimal blocking set of $G-N[r]$. Since there exists an optimum vertex cover of $G$ that contains the vertex set $N(r)$ it follows from Lemma \ref{lemma::nonhereditary_delete} that $|Y|\leq \bsg{G-N[r]} \leq \bsd{\C+c+1}{d-1}$.

        Finally, assume that Case \ref{enum::ub_split_bs_two} (and neither Case \ref{enum::ub_split_bs_r} or Case \ref{enum::ub_split_bs_Nr}) of Lemma \ref{lemma::ub_split_bs} holds. As in the proof of Case 5 of Theorem \ref{theorem::ub_nonhereditary} we bound the size of the minimal blocking set $Y' \subsetneq Y$ of $G-r$ and the minimal blocking set $\hat{Y} = Y \setminus Y'$ of $G-Y'$. Obviously, the size of $Y'$ is at most $\bsd{\C+c}{d-1}$ because $\ed{\C+c}{G-r} = d-1$.
        Recall, the size of $\hat{Y}$ is at most $\bsg{G-Y'-r}-1$ because $\hat{Y}$ is not a blocking set of $G-Y'-r$ whereas $\hat{Y} \cup Z \subseteq \hat{Y} \cup (N(r)\setminus Y')$ is a minimal blocking set of $G-Y'-r$ for a set $Z' \subseteq N(r)\setminus Y'$. Since $Y' \subsetneq Y$ and since $Y$ is a minimal blocking set of $G$, there exists an optimum vertex cover of $G$ that contains the set $Y'$. Thus, we can use Lemma \ref{lemma::nonhereditary_delete} to bound the size of $\bsg{G-Y'-r}$ by $\bsd{\C+c+1}{d-1}$; hence $|\hat{Y}|\leq \bsg{G-Y'-r}-1 \leq \bsd{\C+c+1}{d-1}-1$. We obtain that $|Y|\leq \bsd{\C}{d-1}+\bsd{\C+c+1}{d-1}-1$.
    \end{description}
    Overall, we can bound the size of $Y$ by $\bsd{\C}{d-1}+\bsd{\C+c+1}{d-1}-1$, because $\bsd{\C+c'}{d} \geq 2$ when $c' \geq 1$ (see Observation \ref{observation::bound_bs_bc}). Together with the induction hypothesis we can bound the size of $Y$ as follows:
    \begin{align*}
        |Y| &\leq \bsd{\C+c}{d-1}+\bsd{\C+c+1}{d-1}-1 \\
            &\hspace{-0.5em}\overset{(IH)}{\leq} \hspace{-0.5em} \left(\sum_{i=0}^{d-1} \binom{d-1}{i} f(c+i) - 2^{d-1}+1 \right)
            + \left(\sum_{i=0}^{d-1} \binom{d-1}{i} f(c+1+i) - 2^{d-1}+1 \right) -1 \\
            &= \left(\sum_{i=0}^{d-1} \binom{d-1}{i} f(c+i)\right) + \left(\sum_{i=1}^{d} \binom{d-1}{i-1} f(c+i)\right) - 2 \cdot 2^{d-1}+2 -1 \\
            &= \binom{d-1}{0} f(c) + \sum_{i=1}^{d-1} \left(\binom{d-1}{i}+\binom{d-1}{i-1} \right) f(c+i) + \binom{d-1}{d-1} f(c+d) -2^d +1 \\
            &= \binom{d}{0} f(c) + \sum_{i=1}^{d-1} \binom{d}{i} f(c+i) + \binom{d}{d} f(c+d) -2^{d} +1 \\
            &= \bigg(\sum_{i=0}^{d} \binom{d}{i} f(c+i)\bigg) -2^{d} +1\qedhere
    \end{align*}
\end{proof}

It was shown by Hols and Kratsch that we can bound $\bsg{G}$ for a graph $G$ when $\OPT(G)-\LP(G)=c$ by $2c+2$ \cite[Theorem 14]{HolsK17}.
Let $\C_{\LP} = \{ H \text{ graph} \mid \OPT(H)=\LP(H)\}$ be the set of graphs where the size of an optimum vertex cover equals the value of an optimum \LP solution. Observe that for every graph $G \in \C_{\LP}+c$ it holds that $\OPT(G)-\LP(G) \leq c$: Let $X \subseteq V(G)$ of size at most $c$ such that $G-X \in \C_{\LP}$ (exists by definition of $\C+c$). It holds that $\OPT(G-X)+|X| \geq \OPT(G)$ and that $\LP(G-X)\leq \LP(G)$ which implies that $\OPT(G)-\LP(G) \leq |X|\leq c$.
Thus, $\C_{\LP}$ is a non-hereditary graph class that is $f$-robust with $f(c)=f(\bsc{\C},c)=2c+\bsc{\C_\LP}=2c+2$ .

\begin{corollary} \label{corollary::bound_C_LP}
    It holds that $\bsd{\C_\LP}{d} \leq (\bsc{\C_{\LP}}+d-1)\cdot 2^d +1 = (d+1) \cdot 2^d +1$.
\end{corollary}

\begin{proof}
    This follows directly from Theorem \ref{theorem::ub_nonhereditary}, because $\C_{\LP}$ is a $(2c+\bsc{\C_\LP})$-robust graph class, where $\bsc{\C}=2$ (see above). This implies that
    \begin{align*}
        \bsd{\C_{\LP}}{d} &\leq \sum_{i=0}^d \binom{d}{i} f(\bsc{\C_{\LP}},i) - 2^d+1 \\
                &= \sum_{i=0}^d \binom{d}{i} (2 \cdot i+2) - 2^d+1 \\
                &= (d+1) \cdot 2^d + 1\qedhere
    \end{align*}
\end{proof}

\section{Kernelization results}
In this section, we will combine the results from Sections~\ref{subsec:reducing-num-components} and~\ref{section:bounded-blocking-sets} to obtain polynomial kernelizations for \VC parameterized by a \C-modulator or a $(\C,d)$-modulator.

In Section~\ref{subsubsec:poly-time-rrule-application} we have seen necessary assumptions on a graph class \C such that the number of connected components outside the \C-modulator could be efficiently reduced. We start by extending these results to using $(\C,d)$-modulators.

\begin{lemma} \label{lemma::ed_hereditary_red_rule}
    Let \C be a hereditary graph class on which \VC is polynomial-time solvable.
    In polynomial time we can compute an optimum \VC of a graph $G$ with $\ed{\C}{G}\leq d$, when $d$ is a fixed constant.
\end{lemma}
\begin{proof}
    We prove this via induction over the elimination distance of graph $G$. Obviously, if $\ed{\C}{G}=0$ then $G$ is a graph of the class \C; thus we can compute an optimum vertex cover in polynomial time.

    For the induction step, we assume that we can solve \VC in polynomial time on graphs $G$ with $\ed{\C}{G}<d$.

    Assume that $d=\ed{\C}{G} > 0$. Let $G_1, G_2, \ldots, G_h$ be the connected components of $G$. It is enough to compute the optimum vertex cover of each connected subgraph $G_i$ of $G$ with $1\leq i \leq h$ because $\OPT(G)=\sum_{i=1}^h \OPT(G_i)$. Let $r_i$ be the root of the elimination tree of graph $G_i$. We distinguish between the case that $r_i$ is contained in the optimum vertex cover or not. Hence, $\OPT(G_i) = \min\{\OPT(G_i-r_i)+1, \OPT(G_i - N[r_i])+|N(r_i)|\}$. Observe that the graphs $G_i-r_i$ and $G_i-N[r_i]$ have elimination distance less than $d$ to $\C$: This is clear for $G_i-r_i$ by removing $r_i$ from the elimination tree of $G_i$. For $G_i-N[r_i]$ we can similarly remove $N[r_i]\ni r_i$ from the elimination tree of $G_i$ to see this, using that $\C$ is hereditary. By the inductive assumption we can compute an optimum vertex cover for both graphs in polynomial time.
    Note that while the running time is polynomial for $d$ constant, it may depend exponentially on $d$.
\end{proof}

\begin{corollary}\label{cor:rrule-poly-time-ed-to-hereditary}
    Reduction Rule \ref{rule::delete_cc} is applicable in polynomial time on graphs $G$ with a given $(\C,d)$-modulator $X$, where $\C$ is a hereditary graph class on which \VC is solvable in polynomial time, and where $\bsc{\C}$ is bounded.
\end{corollary}
\begin{proof}
Since \C is hereditary and \VC is solvable in polynomial time on \C, it follows from Lemma~\ref{lemma::ed_hereditary_red_rule} that we can efficiently solve \VC in graphs $G$ with $\ed{\C}{G} \leq d$. Furthermore, it follows from the fact that $\bsc{\C}$ is bounded and Theorem~\ref{theorem::ub_nonhereditary} that $\bsd{\C}{d}$ is bounded. It now follows from Lemma~\ref{lem:rrule-1-in-poly-time-hereditary} that we can apply Reduction Rule \ref{rule::delete_cc} in polynomial time.
\end{proof}
For non-hereditary graph classes \C we again need that \VC is solvable in polynomial time on graph class $\C+c$ with $c$ constant. Here $\C+1$ is not enough because in each recursive step we add a vertex to some of the base components.

\begin{lemma} \label{lemma::ed_nonhereditary_red_rule}
    Let \C be a robust graph class on which \VC is polynomial-time solvable. Furthermore, assume that \VC is polynomial-time solvable on graph class $\C+c$ for every constant $c \in \mathbb{N}$.
    In polynomial time we can compute an optimum \VC of a graph $G$ with $\ed{\C}{G}\leq d$.
\end{lemma}
\begin{proof}
    Again, we use induction to prove the lemma. The construction is similar to the construction of Lemma \ref{lemma::nonhereditary_delete}. For the base case assume that $\ed{\C}{G} = 0$. This implies that $G$ is a graph of the class \C; hence we can compute an optimum vertex cover in polynomial time.

    For the induction step, we assume that we can solve \VC in polynomial time on graphs $G$ with $\ed{\C}{G}<d$ for all graph classes \C that fulfill the requirements of the lemma.

    Let $G$ be a graph with $\ed{\C}{G} = d > 0$. Let $G_1,G_2,\ldots, G_h$ be the connected components of $G$. Again, it is sufficient to compute the optimum vertex cover of each connected subgraph separately. Let $G_i$ be a connected component of $G$ and let $r_i$ be the root of the elimination tree of $G_i$. As before, we want to compute an optimum vertex cover of $G_i-r_i$ and $G_i-N[r_i]$. Since $r_i$ is the root of the elimination tree of $G_i$ it holds that $\ed{\C}{G_i-r_i} = d-1$. Hence, we can compute an optimum vertex cover of $G_i-r_i$ in polynomial time for $d$ constant. To compute an optimum vertex cover of $G_i-N[r_i]$ we construct the graph $\hat{G}_i$ as in the proof of Lemma \ref{lemma::nonhereditary_delete}. Let $Z' \subseteq N(r_i)$ the set of vertices in $N(r_i)$ that are contained in the base components of the elimination tree of $G_i$. It holds that $\ed{\C+1}{\hat{G}_i}=d-1$. Observe that $\OPT(\hat{G}_i)\leq \OPT(G_i-N[r_i])+|Z'|$ because every optimum vertex cover of $G_i-N[r_i]$ together with the set $Z'$ is a vertex cover of $\hat{G}_i$. If $N(r_i)$ is not a blocking set of $G_i$ then $\OPT(G_i-N[r_i]) = \OPT(\hat{G}_i) - |Z'|$ (see proof of Lemma \ref{lemma::nonhereditary_delete}). If $N(r_i)$ is a blocking set of $G_i$ then there exists an optimum vertex cover of $G_i$ that contains $r_i$.
    Furthermore, it holds that $\OPT(\hat{G}_i) - |Z'|+|N(r_i)| \geq \OPT(G_i)$ when $N(r_i)$ is a blocking set of $G_i$: Let $\hat{S}$ be an optimum vertex cover of $\hat{G}_i$. First, we add all vertices in $N(r_i) \setminus Z'$ to $\hat{S}$ and denote the resulting set by $S'$. Recall that $Z' \subseteq N(r_i)$ which implies that $|S'| = |\hat{S}| + |N(r_i)|-|Z'|$. If $Z' \subseteq \hat{S}$ then $S'$ is a vertex cover of $G_i$ that contains $N(r_i)$. This implies that $|S'|> \OPT(G_i)$ because $N(r_i)$ is a blocking set of $G_i$. Since $|S'|=|\hat{S}| + |N(r_i) \setminus Z'|$ we obtain that $\OPT(G_i) < \OPT(\hat{G}) - |Z'| + |N(r_i)|$. If $Z' \nsubseteq \hat{S}$ then $\hat{S}$ contains some newly added vertices. Hence, $S=S' \cap V(G_i) \cup \{r\}$ is a vertex cover of $G_i$. It holds that $\OPT(G_i) \leq |S| = |S' \cap V(G_i)| + |\{r\}| \leq |S'| = |\hat{S}| + |N(r_i) \setminus Z'| = \OPT(\hat{G}) - |Z'|+|N(r_i)|$.

    Overall, this implies that $\OPT(G_i) = \min \{\OPT(G_i-N[r_i])+|N(r_i)|, \OPT(G_i-r_i)+1\} = \min \{\OPT(\hat{G}_i)-|Z'|+|N(r_i)|, \OPT(G_i-r_i)+1\}$. As mentioned above, it holds that $\ed{\C+1}{\hat{G}_i}=d-1$. Thus, we can compute an optimum vertex cover of $\hat{G}_i$ in polynomial time, because the graph class $\C+1$ fulfills the desired properties of the lemma if \C fulfills these properties.
\end{proof}
For the next corollary, observe that in particular $\bsd{\C}{d}$ is bounded if \C is known to be either hereditary or $f$-robust, by Theorems~\ref{theorem::ub_hereditary} and~\ref{theorem::ub_nonhereditary}.

\begin{corollary}\label{cor:rrule-poly-time-ed-to-non-hereditary}
    Reduction Rule \ref{rule::delete_cc} is applicable in polynomial time on graphs $G$ with a given $(\C,d)$-modulator $X$, where $\C$ is a robust graph class with the properties that $\bsd{\C}{d}$ is bounded for any constant $d$ and that \VC is polynomial-time solvable on graph class $\C+c$ for constant~$c$.
\end{corollary}
\begin{proof}
        It follows from Lemma~\ref{lemma::ed_nonhereditary_red_rule} that we can solve \VC in polynomial time on graphs $G$ with $\ed{\C}{G}\leq d$ for any constant $d$.
        It then follows from Lemma~\ref{lemma::non_hereditary_verify_bs} that we can verify whether a given set of vertices is a blocking set in $G$  with $\ed{\C}{G}\leq d$, as adding a single vertex to $G$ increases its elimination distance to \C by at most one.
        The statement now follows from Lemma \ref{lemma::red_rule_poly_time}.
\end{proof}

\subsection{General results}

In this section, we show that \VC parameterized by the size of a $(\C,d)$-modulator has a polynomial kernel when the graph class \C fulfills some additional properties. The assumptions that $\bsd{\C}{d}$ is bounded and that \VC is polynomial-time solvable on the considered graph class are necessary, if these assumptions fail a polynomial kernel is unlikely to exist. The same holds for the assumption that \VC parameterized by a \C-modulator has a polynomial kernel. We additionally require that \C is a robust graph class that is either hereditary, or has the property that \VC is polynomial-time solvable on $\C+c$. These assumptions will ensure that our reduction rule can be applied in polynomial time.

\begin{lemma} \label{lemma::kernel_ed}
  Let $\C$ be a robust graph class for which $\bsd{\C}{d}$ is bounded and on which \VC is polynomial-time solvable, such that \C is hereditary or \VC is polynomial-time solvable on $\C + c$ for all constants $c$.

     Suppose $\VC$ parameterized by the size of a $\C$-modulator $\hat{X}$ has a (randomized) polynomial kernel with $g(|\hat{X}|)$ vertices.
    Then \VC parameterized by the size of a $(\C,d)$-modulator $X$ has a (randomized) polynomial kernel with $\Oh( g( |X|^b ))$  vertices, where $b = \prod_{i=1}^d \bsd{\C}{i}$.
\end{lemma}

Our kernelization for \VC parameterized by the size of a $(\C,d)$-modulator is similar to the kernelization for \VC parameterized by the size of a $d$-treedepth modulator (see~\cite{BougeretS17}). One difference is that we do not want to introduce hyper-edges. For completeness, we give a short proof of Lemma \ref{lemma::kernel_ed}.

\begin{proof}
    Like Bougeret and Sau \cite{BougeretS17} we reduce an instance $(G,k,X)$ of \VC parameterized by the size of a $(\C,d)$-modulator to an instance $(G',k',X')$ of \VC parameterized by the size of a $(\C,d-1)$-modulator. The bound on the number of vertices follows inductively.

    We start by observing that Reduction Rule~\ref{rule::delete_cc} can be applied in polynomial time. We do a case distinction. If \C is hereditary, it follows from Corollary~\ref{cor:rrule-poly-time-ed-to-hereditary} that Reduction Rule~\ref{rule::delete_cc} can be applied in polynomial time.

    Otherwise, \VC is polynomial-time solvable on graphs from $\C + c$ for any constant $c$ and it follows immediately from  Corollary~\ref{cor:rrule-poly-time-ed-to-non-hereditary} that Reduction Rule~\ref{rule::delete_cc} can be applied in polynomial time.

    To obtain the kernel, we first apply Reduction Rule \ref{rule::delete_cc} to instance $(G,k,X)$ of \VC parameterized by the size of a $(\C,d)$-modulator. This leads to an equivalent instance $(G',k',X)$ of \VC parameterized by the size of a $(\C,d)$-modulator where the number of connected components in $G'-X$ is at most $|X|^{\bsd{\C}{d}}$ (Theorem \ref{theorem::bound_numb_cc}). Let $X_r$ be the set of roots of the elimination forest of $G'-X$; hence $X_r$ contains at most $|X|^{\bsd{\C}{d}}$ vertices because every connected component of $G'-X$ has exactly one root.
    Now, $(G',k',X\cup X_r)$ is an instance of \VC parameterized by the size of a $(\C,d-1)$-modulator. Obviously, $(G,k,X)$ is a yes-instance if and only if $(G',k',X\cup X_r)$ is a yes-instance.

    It follows inductively that we can reduce the instance $(G',k',X\cup X_r)$ of \VC parameterized by the size of a $(\C,d-1)$-modulator to an instance $(\hat{G},\hat{k},\hat{X})$ of \VC parameterized by the size of a $(\C,0)$-elimination distance modulator with
    \[
    |\hat{X}| \in \Oh \left(|X \cup X_r|^{\prod_{i=1}^{d-1} \bsd{\C}{i}} \right) = \Oh \left(\left({|X|^{\bsd{\C}{d}}}\right)^{\prod_{i=1}^{d-1} \bsd{\C}{i}} \right) = \Oh \left(|X|^{\prod_{i=1}^{d} \bsd{\C}{i}} \right).
    \]
    Since $(\hat{G},\hat{k},\hat{X})$ is also an instance of \VC parameterized by the size of a $\C$-modulator, we can reduce instance $(\hat{G},\hat{k},\hat{X})$ to an equivalent instance $(\widetilde{G},\widetilde{k},\widetilde{X})$ of \VC parameterized by the size of a $\C$-modulator with $\Oh(g ( |X|^b ) )$ vertices where $b=\prod_{i=1}^d \bsd{\C}{i}$. This concludes the proof.
\end{proof}

Observe that in the above lemma statement, when \VC parameterized by a \C modulator allows a polynomial kernel, the fact that \VC is solvable in polynomial time on graphs from $\C+c$ is immediate: since the problem has a polynomial kernel, it must be FPT in the parameter. Since in this case the size of a $\C$-modulator is $c$, which is constant, the result follows.

In the above theorem statement, we assume that $\bsd{\C}{d}$ is bounded to obtain the kernelization. We observe that for hereditary graph classes, this assumption is not needed, it follows from our results in Theorem~\ref{theorem::ub_hereditary} that it suffices to bound $\bsc{\C}$. Furthermore, a bound on $\bsc{\C}$ often comes naturally: if \VC parameterized by a \C-modulator has a polynomial kernel, it follows from Theorem~\ref{theorem:lb:intro} that, unless \containment, there must exist a constant $d$ such that $\bsc{\C} \leq d$.

\begin{thm:theorem5}
\theoremfive
\end{thm:theorem5}
\begin{proof}
The result is immediate from Lemma~\ref{lemma::kernel_ed}, combined with the bound on $\bsd{\C}{d}$ for hereditary graph classes provided in Theorem~\ref{theorem::ub_hereditary}.
\end{proof}

Similarly, for non-hereditary  graph classes, it suffices if $\C$ is $f$-robust to obtain a polynomial kernel. The size of the kernel depends on $f$.

\begin{corollary} \label{theorem::kernel_ed_hereditary}
Let \C be a robust and $f$-robust graph class for which $\beta_C$ is bounded and for which \VC parameterized by the size of a \C modulator $\hat{X}$ has a (randomized) polynomial kernel. Then \VC parameterized by the size of a $(\C,d)$-modulator has a (randomized) polynomial kernel.
\end{corollary}
\begin{proof}
The result is immediate from Lemma~\ref{lemma::kernel_ed}, combined with the bound on $\bsd{\C}{d}$ for $f$-robust graph classes provided in Theorem~\ref{theorem::ub_nonhereditary}.
\end{proof}

\subsection{Kernel for modulator to bounded $\C_{\LP}$ elimination distance}

In this section, we show how Theorem~\ref{thm:intro:lp-kernel} follows from the general results in the previous section, to have an explicit example for a non-hereditary base class $\C$. That is, we show how to get a randomized polynomial kernel for \VC parameterized by the size of a modulator $X$ such that $G-X$ has bounded elimination distance to the non-hereditary class $\C_{\LP}$ of graphs where integral and fractional vertex cover size coincide.
Towards proving this result, we show the following relation between the value of $\ell = \OPT(G)-\LP(G)$ and the size of a $\C_{\LP}$ modulator in $G$.

\begin{lemma} \label{lemma::bound_modulator_OPT_LP}
    Let $G$ be a graph, and let $\ell = \OPT(G)-\LP(G)$.
    There exists a vertex set $X \subseteq V(G)$ of size at most $2 \ell$ such that $\OPT(G-X)=\LP(G-X)$.
\end{lemma}

\begin{proof}
    Let $x \in \{0,\frac{1}{2},1\}^{|V(G)|}$ be an optimum half-integral solution to the \VC \LP of $G$, and let $V_i = \{ v \in V(G) \mid x_v = i\}$ for $i \in \{0,\frac{1}{2},1\}$. Due to a result of Nemhauser and Trotter \cite{NemhauserT75} there exists an optimum vertex cover $S$ of $G$ with $V_1 \subseteq S \subseteq V_{\frac{1}{2}} \cup V_1$.
    Consider the bipartite graph $H$ where one part is the set $V_\frac{1}{2} \cap S$ and the other part is the set $V_\frac{1}{2} \setminus S$, and where there is an edge between $y \in V_\frac{1}{2} \cap S$ and $z \in V_\frac{1}{2} \setminus S$ if and only if $\{y,z\} \in E(G)$. Let $M$ be a maximum matching in the bipartite graph $H$.
    It holds that the matching $M$ saturates the set $V_\frac{1}{2} \setminus S$; otherwise it follows from Theorem \ref{theorem::matching} that there exists a set $Z \subseteq V_\frac{1}{2} \setminus S$ such that $|N_H(Z)|<|Z|$. But, this implies that $x$ is not an optimum half-integral solution to the \VC \LP of $G$ because $x'$ with $x_v'=1$ for all $v \in V_1 \cup N_H(Z)$, $x_v'=\frac{1}{2}$ for all $v \in V_\frac{1}{2} \setminus N_H[Z]$, and $x_v'=0$ for all $v \in V_0 \cup Z$ is also a valid $\LP$ solution to the \VC \LP of $G$ with $|x'|<|x|$. (Note that $Z$ is an independent set because it is disjoint from $S$ and that $N_G(Z)\subseteq V_{\frac{1}{2}} \cup V_1$ because $Z\subseteq V_\frac{1}{2}$.) Thus, the matching $M$ saturates $V_\frac{1}{2} \setminus S$.

    Let $X= (V_\frac{1}{2} \cap S) \setminus V(M)$ be the set of vertices in $V_\frac{1}{2} \cap S$ that are not an endpoint of a matching edge of $M$.
    We will show that $|X|$ has size $2 \ell$. It holds that $\OPT(G)=|S|=|V_\frac{1}{2} \cap S| + |V_1|$ and that $\LP(G) = |V_1| + \frac{1}{2} |V_\frac{1}{2}| = |V_1| + \frac{1}{2} |V_\frac{1}{2} \cap S| + \frac{1}{2} |V_\frac{1}{2} \setminus S|$. Since $\ell = \OPT(G)-\LP(G)$ we obtain that
    \begin{align*}
        \ell &= \OPT(G)-\LP(G) = |V_\frac{1}{2} \cap S| + |V_1| - \left( |V_1| + \frac{1}{2} |V_\frac{1}{2} \cap S| + \frac{1}{2} |V_\frac{1}{2} \setminus S| \right) \\
            &= \frac{1}{2} |V_\frac{1}{2} \cap S| - \frac{1}{2} |V_\frac{1}{2} \setminus S|
            = \frac{1}{2} |(V_\frac{1}{2} \cap S)\setminus V(M)|
            = \frac{1}{2} |X|.
    \end{align*}
    The second to last equality holds because the matching $M$ saturates $V_\frac{1}{2} \setminus S$. It follows that $|X|=2 \ell$.
    Observe that the graph $G-X$ has a matching of size $|V_1|+|M|$ and that $S \setminus X$ is a vertex cover of $G-X$ of size $|V_1|+|M|$. This implies that $\OPT(G-X)=\LP(G-X)$, and concludes the proof.
\end{proof}

Combining Corollary \ref{corollary::bound_C_LP} and Lemma \ref{lemma::kernel_ed} we can now generalize the kernelization for \VC parameterized by the size of a $d$-treedepth modulator and parameterized by the difference between an optimum vertex cover and an optimum \LP solution using the size of a $(\C_\LP,d)$-modulator as the parameter.
The following theorem subsumes Theorem~\ref{thm:intro:lp-kernel}.

\begin{theorem} \label{theorem::combine_kernels}
    An optimum $(\C_\LP,d)$-modulator of a graph $G$ has at most the size of a $d$-treedepth modulator of $G$ and at most twice the size of $\OPT(G)-\LP(G)$.
    Furthermore, \VC parameterized by the size of a $(\C_\LP,d)$-modulator admits a randomized polynomial kernel.
\end{theorem}
\begin{proof}
    The empty graph is contained in $\C_\LP$ because both integral and fractional vertex cover size are zero. It follows directly that the elimination distance to $\C_\LP$ is upper bounded by the treedepth, i.e., $\ed{\C_\LP}{G}\leq \td{G}$, and that every $d$-treedepth modulator of $G$ is also a $(\C_\LP,d)$-modulator of $G$.

    We showed in Lemma \ref{lemma::bound_modulator_OPT_LP} that there exists a $\C_\LP$-modulator in $G$ of size at most $2 \cdot (\OPT(G)-\LP(G))$. This modulator is also a $(\C_\LP,0)$-modulator of $G$. Hence, the size of an optimum $(\C_\LP,d)$-modulator of $G$ is at most the size of a $\C_\LP$-modulator of $G$ which is at most $2 \cdot (\OPT(G)-\LP(G))$.

    Now, we will show that \VC parameterized by the size of a $(\C_\LP,d)$-modulator admits a polynomial kernel. It holds that $\bsd{\C}{d}$ is bounded (see Corollary \ref{corollary::bound_C_LP}). Furthermore, \VC parameterized by the size of a $\C_\LP$-modulator admits a polynomial kernel because the size of the modulator is at most the difference between an optimum vertex cover and an optimum \LP solution, and because \VC parameterized by the difference between an optimum vertex cover and an optimum \LP solution has a (randomized) polynomial kernel \cite{KratschW12}.

    Furthermore, we can show that \VC is solvable in polynomial time on graphs from $\C_{\LP} + c$. Let $G$ be a graph from $\C_{\LP} + c$ and let $X\subseteq V(G)$ such that $G-X\in\C_{\LP}$ and $|X|=c$. Then
    $\vcopt(G) \leq \vcopt(G-X) + |X|$ and $\LP(G) \geq \LP(G-X)$. Thereby, $\vcopt(G) - \LP(G) \leq \vcopt(G-X) + |X| - \LP(G-X) \leq c$. Since \VC parameterized by $\vcopt(G) - \LP(G)$ is FPT \cite{NarayanaswamyRRS12}, it follows that when $\vcopt(G) - \LP(G)$ is constant, the problem is solvable in polynomial time.

    Thus, it follows from  Lemma \ref{lemma::kernel_ed} that \VC parameterized by the size of a $(\C_\LP,d)$-modulator admits a (randomized) polynomial kernel.
\end{proof}

\section{Conclusion}\label{section::conclusion}

In the first part (Section~\ref{section::relation}) we have showed that bounded minimal blocking set size in \C is necessary but not sufficient to get a polynomial kernel for \VC when parameterized by the size of a modulator $X$ to a robust (or at least union-closed) class \C. We then showed that bounded minimal blocking set size suffices to efficiently reduce the number of components of $G-X$ assuming that $\C$ is robust (so deletion of components lets $G-X$ stay in \C) and that we can efficiently compute optimum vertex covers and test blocking sets in graphs of \C. The obtained bound of $\Oh(|X|^{\bsc{\C}})$ components is likely optimal because it matches the size lower bound proved earlier, which requires only components of constant size.

In the second part we first proved bounds for the minimal blocking set size relative to elimination distances to classes \C, motivated by the bounds that Bougeret and Sau~\cite{BougeretS17} obtained relative to treedepth (Section~\ref{section:bounded-blocking-sets}). We obtain the exact value for all hereditary classes \C and slightly weaker upper bounds for certain non-hereditary classes \C. This enabled new polynomial kernelization results for \VC that effectively replace (the size of) a modulator to a class \C to modulators to graphs of bounded elimination distance to \C, e.g., when $\C$ is the class of forests, bipartite graphs, or $C_{\LP}$ (where integral and fraction vertex cover size coincide).

As future work it would be great to get a similar kernelization result when parameterized by the size of a modulator to bounded elimination distance to the graph class $\C_{2\LP-\MM}$ where $\OPT=2\LP-\MM$ (i.e., minimum vertex cover size equals two times fractional cost minus size of a maximum matching, cf.~\cite{GargP16}), which relates to the randomized kernelization for the corresponding above guarantee parameterization~\cite{Kratsch18}. This would essentially subsume and generalize all currently known polynomial kernelizations for \VC (to which we came close with the result for bounded elimination distance to $\C_{\LP}$). It would also be nice to have tight bounds for the maximum size of minimal blocking sets in the non-hereditary case, and to get such bounds with fewest possible technical assumptions.



\newpage
\bibliography{lit}

\end{document}